\documentclass[conference]{IEEEtran}

\usepackage[utf8]{inputenc}
\usepackage{amssymb}
\usepackage{url}
\usepackage[font=small]{caption}
\usepackage{graphicx}
\usepackage{subcaption}
\usepackage{lipsum}
\usepackage{mwe}
\usepackage{multirow}
\usepackage{xspace}
\usepackage{titlesec}
\usepackage[draft]{changes}
\usepackage{listings}
\usepackage{fancyvrb}
\usepackage{framed}
\usepackage{amsthm}
\usepackage[listings,skins]{tcolorbox}
\usepackage{amsfonts}
\usepackage{amsmath}
\usepackage[normalem]{ulem}
\usepackage{filecontents}
\usepackage[normalem]{ulem}
\usepackage[framemethod=TikZ]{mdframed}
\usepackage{setspace}

\usepackage{pifont}%

\definechangesauthor[color=red]{HM}
\definechangesauthor[color=blue]{K.K.}

\colorlet{dark-green}{green!50!black}
\colorlet{dark-cyan}{cyan!50!black}

\newif\ifshownotes
\shownotestrue
\ifshownotes
\newcommand{\knote}[1] {{$\langle${\textcolor{blue}{K.K.: \textbf{#1}}}$\rangle$}}
\newcommand{\hnote}[1] {{$\langle${\textcolor{dark-cyan}{HM: \textbf{#1}}}$\rangle$}}
\newcommand{\jnote}[1] {{$\langle${\textcolor{dark-green}{Jia: \textbf{#1}}}$\rangle$}}
\newcommand{\snote}[1] {{$\langle${\textcolor{orange}{Silas: \textbf{#1}}}$\rangle$}}
\else
\newcommand{\knote}[1] {\xspace}
\newcommand{\hnote}[1] {\xspace}
\newcommand{\jnote}[1] {\xspace}
\newcommand{\snote}[1] {\xspace}
\fi

\newcommand{\hide}[1]{}
\newcommand{\eg}{{\em e.g., }}

\newcommand{\vs}{{\em vs. }}

\newcommand{\calU}{\mathcal{U}}
\newcommand{\msg}{{\sf msg}}
\newcommand{\calA}{\mathcal{A}}
\newcommand{\calB}{\mathcal{B}}
\newcommand{\calC}{\mathcal{C}}
\newcommand{\pke}{{\sf PKE}}
\newcommand{\calN}{\mathcal{N}}
\newcommand{\pk}{{\sf pk}}
\newcommand{\G}{\mathbb{G}}
\newcommand{\bbB}{\mathbb{B}}
\newcommand{\veb}{{\sf VEB}}
\newcommand{\bM}{{\bf M}}

\newtheorem{thm}{Theorem}

\newtheorem{clm}{Claim}
\newcommand{\abe}{{\sf ABE}}

\newcommand{\setup}{{\sf Setup}}

\newcommand{\keygen}{{\sf KeyGen}}
\newcommand{\enc}{{\sf Enc}}
\newcommand{\dec}{{\sf Dec}}

\newcommand{\mpk}{{\sf mpk}}
\newcommand{\msk}{{\sf msk}}

\newcommand{\sk}{{\sf sk}}
\newcommand{\ct}{{\sf ct}}

\newcommand{\sfspan}{{\sf span}}

\titlespacing*{\section}
{0pt}{4pt plus 1pt minus 4pt}{4.3pt plus 4pt}
\titlespacing*{\subsection}
{0pt}{4pt plus 1pt minus 4pt}{4.3pt plus 4pt}
\makeatletter
\renewenvironment{itemize} {
\begin{list}{$\bullet$}
    {
    \setlength{\itemsep}{2pt}
     \setlength{\parsep}{2pt}
     \setlength{\topsep}{2pt}
     \setlength{\partopsep}{0pt}
     \setlength{\leftmargin}{0.9em}
     \setlength{\labelwidth}{0.6em}
     \setlength{\labelsep}{0.4em}
    }
}
{\end{list}}
\makeatother

\makeatletter
\def\@IEEEsectpunct{\ \,}
\def\paragraph{\@startsection{paragraph}{4}{\z@}{1.5ex plus 3ex minus 3ex}%
{-.01em}{\normalfont\normalsize\bfseries}}
\makeatother

\newtheoremstyle{noindentdefn}   %
  {3pt}                          %
  {3pt}                          %
  {\itshape}                    %
  {}                            %
  {\bfseries}                  %
  {}                           %
  {.5em}                        %
  {\thmname{#1}~\thmnumber{#2} \thmnote{ \textnormal{#3}}}                            %

\theoremstyle{noindentdefn}

\begin{document}

\newcommand{\name}{\textit{ABEAT}\xspace}
\title{\name: Efficient and Anonymous Encryption for ABE-based Dynamic Group Communication\vspace{-5mm}}

\author{\IEEEauthorblockN
{Hongmiao Yu\IEEEauthorrefmark{1}, Silas Richelson\IEEEauthorrefmark{1}, Jiachen Chen\IEEEauthorrefmark{2}, K. K. Ramakrishnan\IEEEauthorrefmark{1}}
\IEEEauthorblockA{\IEEEauthorrefmark{1}
Department of Computer Science and Engineering, 
University of California Riverside, Riverside, USA}
\IEEEauthorblockA{\IEEEauthorrefmark{2} WINLAB, 
Rutgers University, North Brunswick, NJ, USA.}

\vspace{-13mm}}

\thispagestyle{plain}
\pagestyle{plain}
\maketitle

\begin{abstract}

Confidential communication among a dynamic group of participants that  ensures flexible and efficient many-to-many communication is highly desired capability.  
We leverage attribute-based encryption (ABE) for confidential group communication and enhance it by a graph-based namespace to create an efficient framework that allows groups to be formed and changed dynamically. 
In this paper, we focus on the important additional need to maintain the anonymity of recipients of a message, when using ABE for group communication for a variety of usage scenarios (e.g., emergency response). 

We propose \name, an efficient and anonymous dynamic group communication system that also minimizes overhead on receivers who are not the intended recipients of a message. 
In \name, we propose a new anonymous KP-ABE approach to maintain recipient anonymity. \name hides the clear attribute in the ciphertext of KP-ABE, and also prevents several attacks that seek to break anonymity. 
\name provides fast recipient verification, dramatically lowering  the decryption overhead for non-recipients by more than a factor of 90 versus the current state of the art such as hidden vector encryption (HVE). In fact, it is 
even 40\% less than FABEO, which offers no anonymity.
 
\end{abstract}

\section{Introduction}
Group communication is critical to our society, being utilized for immediate interactive communication and notification among people and organizations. While commercial platforms such as WhatsApp and Signal provide large-scale solutions for secure communication, many specialized use cases require additional features which are not provided by the big platforms.  Consider, for example, situations such as emergency response, military tactical operations, or corporate communications.  Since group management (creation and modification) is costly for these large platforms, they are problematic in situations where the group composition changes dynamically.  Additionally, end-to-end encryption (usually considered the ``gold standard'' for security in group communication) actually precludes features such as auditing which might be desirable. Indeed, during an emergency, first responders from different departments might collaborate, forming task-oriented groups for specific missions.  Membership of these groups may need to change on the fly, without requiring costly group management operations every time.  Moreover, while communication transcripts must be encrypted to ensure confidentiality against malicious parties, they would also need to be accessible to internal departments to ensure accountability, and allow for post-incident forensic analysis. Another requirement is that the underlying network infrastructure during emergencies may experience disruptions, resulting in loss, delay, and  
messages being delivered out-of-order, rather than delivering messages (especially application control messages) reliably.

{\bf Use of Namespaces for Groups. }Previous work has shown that namespaces are an effective mechanism for managing dynamic group communication for emergency response~\cite{POISE, SAFE}. Namespaces facilitate the dynamic formation of groups among recipients and are typically used in a publish/subscribe framework structured according to the role-hierarchy of the participating members. Before an operation begins, members subscribe to names that correspond to their assigned roles. During the operation, a sender publishes messages to one or more roles rather than to specific individuals. 
Furthermore, namespaces can be flexibly modified, or names (or even sub-trees of the namespace) added during an emergency response, supporting highly dynamic group membership changes. As described in ~\cite{SAFE}, Attribute-Based Encryption (ABE) can be an effective mechanism to ensure confidentiality for such namespace-based dynamic group communication. ABE provides several advantages: (1) No need for group key establishment: ABE eliminates the requirement for a shared group key among members. This is especially valuable in emergency or tactical scenarios where some participants may be unreachable intermittently due to damaged infrastructure~\cite{POISE}. By treating namespaces as attributes, senders can target specific roles at the time of sending a message, freely adjusting the recipients for each transmission. In contrast, group messaging systems such as Message Layer Security (MLS) require ordered, lossless message delivery, and relatively stable membership, making them unsuitable for highly dynamic groups and unreliable communication environments. (2) Fine-grained access control: ABE supports flexible policies  
with attributes, enabling mechanisms such as time-based access control (e.g., using timestamps as attributes to determine when a message becomes accessible). Other namespace-based encryption approaches that rely on traditional PKE—such as NDN trust schemes~\cite{TrustScheme}—lack this level of expressiveness.
(3) Centralized key issuance for accountability: ABE relies on a central authority to issue decryption keys. While centralized key management may raise privacy concerns in general-purpose consumer focused messaging applications, it is advantageous in specialized operations (first responder or corporate group communications). The key issuer can support auditing, oversight, and post-operation analysis by appropriate entities with authorized access.

{\bf Anonymous ABE.} Anonymity is a security notion which guarantees that nobody is able to tell who a message is intended for, unless it has been sent to them.  Since ABE attributes are typically not hidden by ABE ciphertexts, our plan to use the attributes to identify the target recipients raises a security concern which would be extremely problematic in an emergency response scenario.  For example, imagine a hostage situation: if law enforcement is communicating non-anonymously, then the perpetrators might notice if the task force HQ sends a series of messages to a particular unit.  This could reveal the rescue plan to perpetrators, \emph{even if they are unable to read the actual messages}.

Prior work has looked at strengthening ABE security by incorporating various notions of anonymity, and existing solutions can be broadly classified into two types.  There are \emph{general purpose solutions} which use exotic cryptographic primitives to provide full anonymity in ABE~\cite{HVE,HideCP19,HCP15,HCP16}.  However, the resulting schemes are very inefficient for practical use; \emph{e.g.}, they are orders of magnitude slower than state-of-the-art ABE without anonymity, as we show in our evaluations(\S\ref{sec:eval}). 
There are also \emph{application-specific solutions} that sacrifice full anonymity to meet application-specific requirements in exchange for higher efficiency~\cite{FEASE,PHCP18}. A practical design that achieves both full anonymity and high efficiency is therefore needed.

\paragraph*{Our Contributions.} In this work, we premier \name $-$ ABE that is Anonymous and Tunes out non-recipients $-$ which enables lightweight anonymous group communication for a variety of situations involving dynamic groups. \name is effective even when messages may not be delivered in sequence or in a timely manner, and end-to-end communication may be disrupted, such as during a disaster response.  Specifically, \name is an ABE-based group communication scheme with the following attractive features.

\begin{itemize}
    \item[$\bullet$ {\bf Seamless Group Management:}] \name harnesses the ``one-to-many'' encryption capability of ABE to make group management completely trivial.  \name has no explicit group formation, using instead self-describing per-message groups.  In other words, for each message the sender simply specifies the group of intended recipients during encryption.  This stands in contrast to typical group communication systems where group management involves a non-trivial cost in terms of control plane messaging (whether it uses multicast or is server-based), especially when recipient anonymity is desired. 

    \item[$\bullet$ {\bf Recipient Anonymity:}] The most novel security feature of \name is recipient anonymity, meaning that messages sent across the network in \name do not reveal who they are intended for.  Our solution is efficient and ``application specific'' (in the context of the previous discussion), taking advantage of the fact that in our setting, the attributes we wish to hide correspond to users who are registered into our system before communicating.  By inserting a special key generation procedure into the registration process, we achieve anonymity for essentially the same cost as standard (non-anonymous) ABE.

    \item[$\bullet$ {\bf Recipient Verification:}] One unintended side effect of introducing recipient anonymity, is that users cannot tell whether a particular message is intended for them, other than by attempting decryption.  Obviously, requiring every user to decrypt every message (whether that user is a recipient of the message or not) would constitute a massive waste of system resources, and would be untenable.  Instead, \name includes a novel recipient verification module so users can inexpensively tell whether a message is meant for them, without performing the expensive ABE decryption. This check is roughly 80x less costly than full decryption.

    \item[$\bullet$ {\bf Dual authorities:}] Many modern group chat applications provide ``end-to-end encryption'' so that messages sent using the platform can only be decrypted by the intended receiver; in particular, the platform itself cannot decrypt and read the messages. One issue with using ABE for group messaging is that ABE systems produce a master secret key.  This key is typically held by a central key authority who can use the master key to decrypt all messages.  \name instead splits the central key authority into two (or more) ``dual key authorities'', each of which manages only a share of the master secret key.  \name provides end-to-end encryption as long as the authorities do not collude with one another.  Moreover, \name enables message auditing if and when the authorities cooperate to decrypt selected messages (e.g., for an audit of first responder operations).    

    \item[$\bullet$ {\bf Agnostic of Infrastructure:}] \name can use any form of underlying message delivery, including broadcast.  The cost of encryption for each party in \name is not substantially higher than the standard one-to-one public key encryption.

    \item[$\bullet$ {\bf Fine-Grained Access Control:}]  As an ABE scheme, \name is compatible with prior work in ABE-based group communication~\cite{GroupABE1,GroupABE2,GroupABE3,SAFE}.  In particular, \name supports the efficient revocation procedure of~\cite{SAFE}.

    \item[$\bullet$ {\bf Working Dynamicity:}] \name supports dynamic groups by utilizing publish/subscribe on namespaces (see Fig.~\ref{fig:namespace} for an example). We discuss the details in~\ref{sec:overview_model}. 
\end{itemize}
\vspace{-2mm}
\paragraph*{Evaluation} Our proposed anonymous ABE takes around 211.47ms to finish encryption for 100 attributes (recipients), while the decryption time is around 35.32ms. In contrast, HVE~\cite{hve08} takes 4480 ms for encryption and 68 ms for decryption. These experiments conservatively limited the total number of users (in the namespace) to 2000.  Moreover, our efficient recipient verification module cuts decryption time for unintended recipients by roughly 80×.

\section{Background and Related Work}
\label{sec:overview}

\paragraph*{Group Communication using Namespaces} 
Users 
in a group communication environment may be organized into a namespace structured as a directed acyclic graph~\cite{POISE}.  Each node in the graph has a name, and an edge between names indicates a heredity relationship (Fig.~\ref{fig:namespace}).  
The use of namespaces (whether as a strict tree or a more general graph) has been adopted in numerous domains (Domain Name System (DNS)~\cite{DNS}, Information Centric Networks (ICN)~\cite{NDN, POISE} etc.). Using a namespace framework naturally synchronizes dynamic group creation (\emph{e.g.}, when a team of first responders is `dispatched' to an emergency~\cite{POISE}) and other changes across the system by publishing the updated namespace to all subscribers.  

\begin{figure*}[t]
\begin{minipage}[b]{.4\linewidth}
    \captionsetup{justification=centering}
    \centering
    \includegraphics[width=\linewidth]{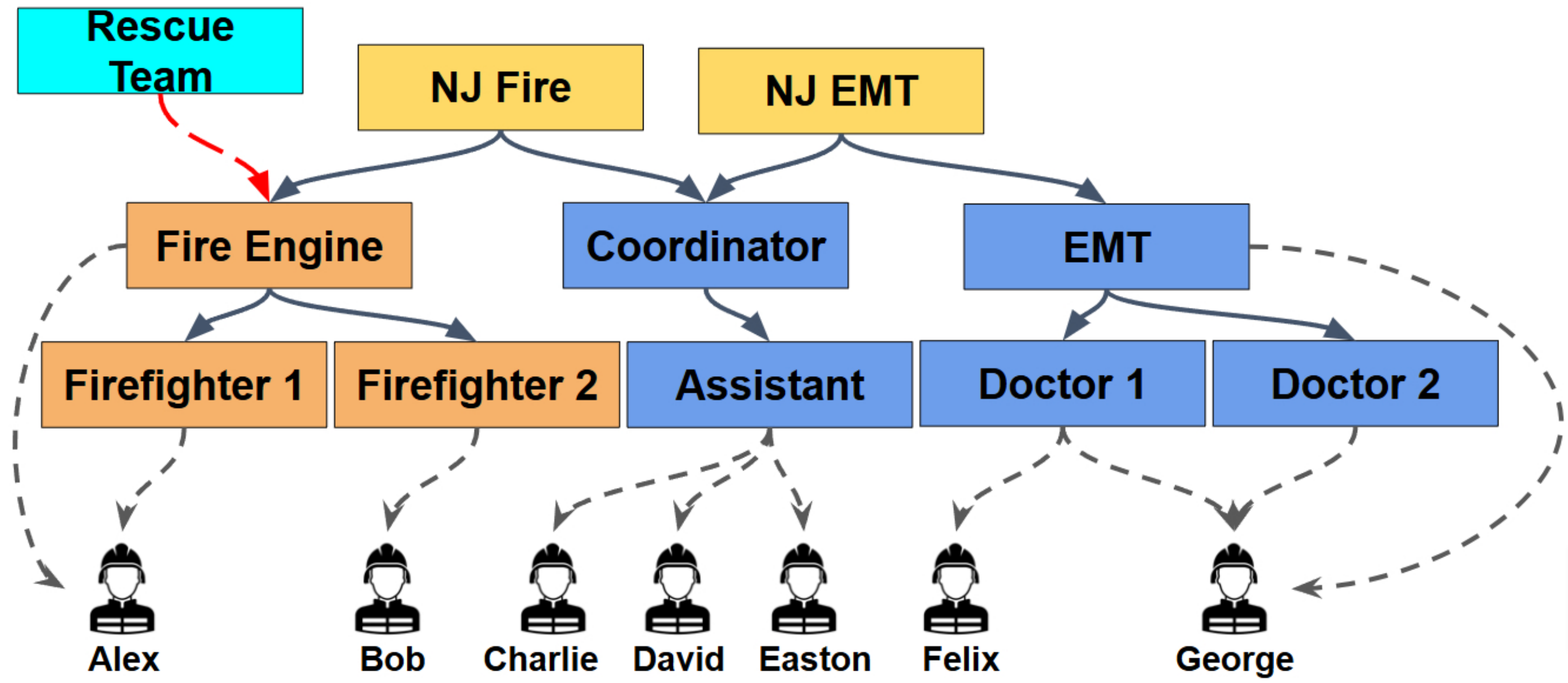}
    \captionof{figure}{{\bf A Namespace Graph for First Responders}}
    \vspace{-2mm}
    \label{fig:namespace}
\end{minipage}\hfill
\begin{minipage}[b]{.6\linewidth}
    \captionsetup{justification=centering}
    \centering
    \includegraphics[width=\linewidth]{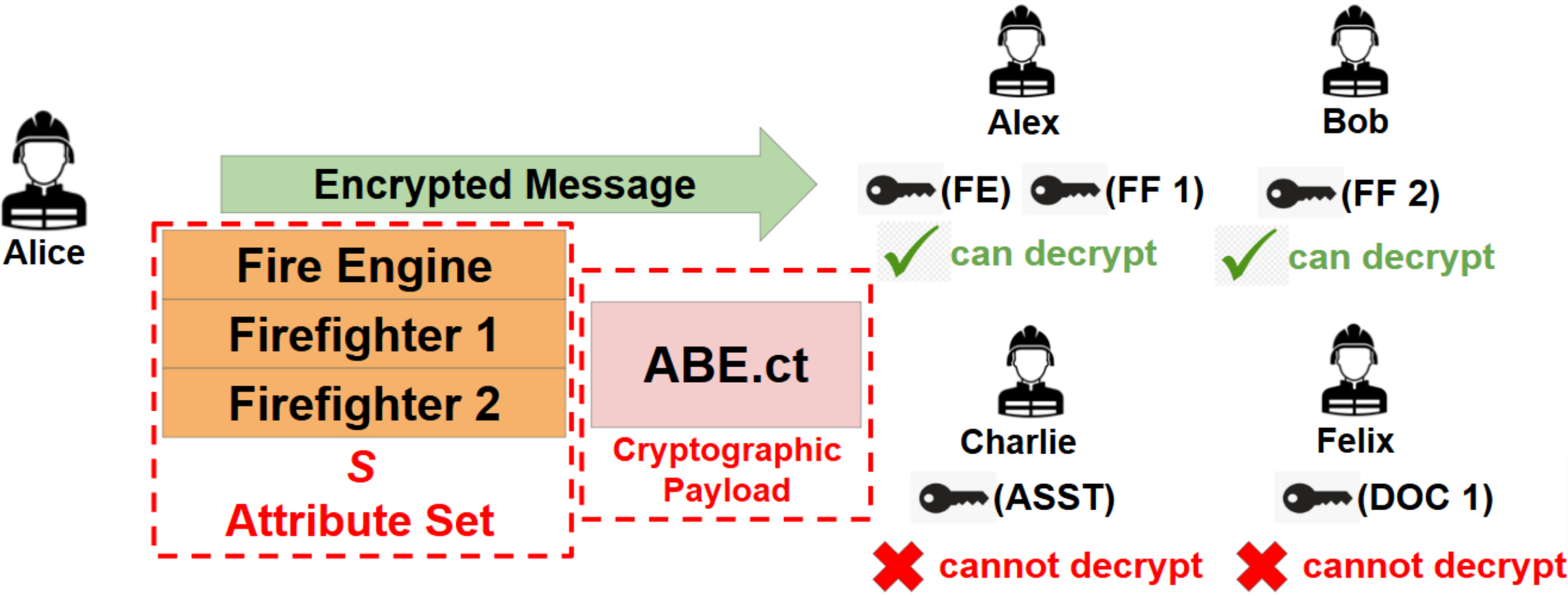}
    \captionof{figure}{ {\bf KP-ABE Group Communication utilizing Namespace}} 
    \label{fig:ABE_group}
    \vspace{-2mm}
\end{minipage}
\end{figure*}

\paragraph*{Dynamicity in group communication:}
The requirement that users interface with ABE via the namespace graph recognizes the varying degrees of ``dynamic'' changes that occur in the system that may impact the ABE scheme. %
\begin{itemize}
    \item \textbf{Changes to the Set of Intended Recipients:} 
    When a group of communicating users need to add a user to their group, this is convenient, as \emph {every message} in our system specifies its own set of intended recipients. %
    A sender simply adds the user to the list of intended recipients of each new message.
    
    \item \textbf{Changes to the Set of Users Subscribing to a Name:} Another type of change occurs when a user subscribes to a new name.  The user simply receives the secret key from the key authorities for that corresponding name the user subscribes to. The user can then decrypt all messages addressed to that name. No other party needs to be involved.
    
    \item \textbf{Changes to the Namespace Graph:} When the namespace is modified, \eg when a new name hierarchy is formed (\eg when a new incident response group of first responders is created for a particular disaster), then the namespace needs to be updated for all users. This enables potential senders (any user) to send to a group identified by the new name (and all the names below that subtree with the new name as the root). Impact on key generation and distribution is controlled and limited. Only those subscribers in the sub-tree rooted at the node that changes in the namespace will need to receive new/additional keys, as can be expected. No other subscribers or senders require key updates. %
\end{itemize}

\vspace{-1mm}
\paragraph*{Namespace Deployment During Disasters.}
Our earlier studies have explored the deployment of namespaces for
publish/subscribe-based group communication~\cite{POISE,RediCom,CNS,
conice,SAFE}. A higher-level authority can deploy and manage a shared
namespace for first responders across different departments. Devices can then synchronize namespace updates using the regional approach in~\cite{conice}.

\subsection{Related work}
\label{sec:related_work}

\paragraph*{Improvements to ABE} Several recent works have made advances to efficient ABE implemented in bilinear groups~\cite{FAME,FABEO,ABGW17}. Other works focus on adding the guarantee of anonymity to ABE~\cite{HideCP19, zhang2024, yang2016, FEASE, FABESA, HVE} for various applications. These schemes either do not provide any anonymity~\cite{FAME,FABEO,ABGW17}, only partially hide the attributes~\cite{FEASE,FABESA,zhang2024}, or incur high performance costs~\cite{yang2016,HideCP19,HVE}. 

\paragraph*{Multi-Authority ABE} 
Multi-authority ABE limits single-authority compromise by assigning
attribute subsets to separate authorities~\cite{
chase07,lewko11,isabella,RW2013MAABE}. In contrast, \name splits each
receiver's decryption key between two authorities, so neither can
generate a complete key alone. %

\paragraph*{ABE-Based Group Communication}
Previous studies have applied ABE to secure group
communication~\cite{cheung2007collusion,traynor2008realizing,
yu2008attribute,kapadia2007attribute,rao2013recipient,SAFE}.
However, some do not provide recipient anonymity~\cite{
cheung2007collusion,traynor2008realizing,SAFE}, while others require
costly cryptographic operations~\cite{yu2008attribute}, an additional
trusted intermediary~\cite{kapadia2007attribute}, or restricted access
policies~\cite{rao2013recipient}. In contrast, \name combines a lightweight PKE-based method for anonymity and a graph-based namespace for flexible group management.

\paragraph*{Other Solutions for Secure Group Communication} 
Non-ABE Secure group messaging (SGM) and continuous group key agreement (CGKA) protocols, including MLS~\cite{mls,MLSSecurity,CGKA}, use TreeKEM~\cite{TreeKEM1,TreeKEM2,treekempaper} to manage group members. However, maintaining a consistent tree requires ordered and
timely delivery of control updates, making these protocols less suitable for disruption-prone environments.

\section{Preliminaries for \name}

\subsection{Our Model}
\label{sec:overview_model}
\paragraph*{Network Layer's role in \name} Let $\calN$ denote the set of names in the namespace (\emph{i.e.}, the vertices of the namespace graph). Messages in our group communication system are addressed to a set $P\subset\calN$ of \emph{intended recipients}.  We use attribute-based encryption to ensure message confidentiality.  So our system delivers ciphertexts to the intended recipients in $P$, who recover the message via decryption.  We intend \name to be anonymous, namely to hide the set $P$, even from an adversary who is monitoring the network traffic. For this purpose, we intend that messages be delivered to a much larger set of \emph{receivers} $Q\subset\calN$, even though only the intended recipients in $P\subset Q$ will possess the keys necessary to decrypt.  $Q$ serves to obscure $P$ and prevent revealing the intended recipients to anyone monitoring network traffic.

Note that the size of $Q$ represents a trade-off between anonymity and system efficiency.  One extreme would be to let $Q=\calN$ (\emph{i.e.}, broadcast every ciphertext to all users), though this is likely to be inappropriate in many settings (\emph{e.g.}, if there are millions of users) as it is too wasteful.  On the other hand, setting $Q=P$ does not introduce system overhead but, of course, completely reveals the intended recipients to anyone monitoring the network.  Thus, $P\subsetneq Q\subsetneq\calN$ would typically be the case, if the network used some form of multicast (\eg IP multicast or overlay multicast).  This is a common scenario used for secure multiparty communication~\cite{BE, OptimalBE20, OptimalBE22}, where messages intended for selected recipients are encrypted within a larger set of receivers to whom the message is delivered.  Our intention is to decouple the communication mechanism at the network layer (\emph{e.g.}, broadcast, multicast, multiple overlays, or even unicast) from how confidentiality is achieved using ABE in this group setting.  We discuss various choices for $|Q|$ compared to $|P|$ in Section~\ref{sec:eval}.

Finally, we mention here the importance of efficient recipient verification.  Since setting $|Q|$ to be large enables recipient anonymity, an extremely important element of our system is the ability of receivers in $Q$ to efficiently check whether they are an intended recipient in $P$ or not. A receiver would only perform the expensive ABE decryption if they are in $P$.  Without this, the total computation of our system would be dominated by unintended recipients trying and failing to decrypt messages that are not addressed to them. 

\paragraph*{Threat and Security Model for \name}
\label{sec:security_threat}

Our target adversary is one who is able to intercept the network traffic in our system (possibly by sniffing, gaining access to network links, routers or switches) and seeks to identify the intended recipients of a message.  Our goal is to hide the identities of the intended recipients from such an adversary. Naturally, we also require message confidentiality, which is inherently provided by ABE.

We also consider another type of attack where the master key is leaked. Our dual key authority approach, splitting the single authority into two (or more), achieves security even when the master secret key of one of the authorities is leaked.  

Our primary goal is to provide recipient anonymity. We acknowledge that a complete group communication solution should also address sender anonymity and message authenticity. Existing techniques for anonymizing the sender (\eg onion routing~\cite{Tor}) and ensuring message authenticity (\eg digital signatures~\cite{dss}) can be readily integrated with our approach.

Our work builds on top of the efficient
FABEO encryption scheme [16], where security is proved in the generic group model (GGM). We inherit these assumptions and focus on confidentiality and anonymity, excluding message blocking, in-transit modification, and denial-of-service attacks.

\subsection{Group Communication using ABE without Anonymity}
\label{sec:overview_naive}
In this subsection, we describe a scheme that uses ABE to build secure group communication without anonymity, as the simplest building block (shown in Fig.~\ref{fig:ABE_group}). It demonstrates syntactically how ABE connects with the namespace graph for the purposes of achieving group communication. 
For the remainder of this section, we let $\abe.\enc$ denote the encryption procedure of an ABE scheme, and so it takes/gives:
\begin{flushright}
\flushleft{\bf Input:} a tuple $(\mpk,S,\msg)$, where $\mpk$ is the master public key, $S$ is a set of attributes, and $\msg$ is a message;
\flushleft{\bf Output:} a pair $(S,\abe.\ct)$, where $S$ is the attribute set and $\abe.\ct$ is the remainder of the ABE ciphertext which we will call the \emph{cryptographic payload}.
\footnote{This is consistent with the terminology in~\cite{FEASE}.}
\end{flushright}

This syntax is standard for the ``key-policy'' variant of ABE (KP-ABE).  We also remind the reader that the key generation procedure for KP-ABE is handled by a central key-authority who takes an access structure $\mathbb{A}$ as input, and (using the master secret key) outputs a secret key $\sk_{\mathbb{A}}$ for $\mathbb{A}$.  Decryption allows anyone with a $\sk_{\mathbb{A}}$ to recover $\msg$ from $(S,\abe.\ct)$ when $\mathbb{A}(S)=1$ holds. 

Suppose Alice wishes to send a message to a set $P\subset\calN$ of intended recipients in the namespace, she can simply have the attribute set $S$ contain the names in $P$, 
and then can publish \[(S,\abe.\ct)\sim\abe.\enc(\mpk,S,\msg).\]  Any user in the system who subscribes to a name $p\in\calN$ will hold a secret key $\sk_{\mathbb{A}_p}$ for the access structure $\mathbb{A}_p$ that takes a set $S\subset\calN$ (the attribute should be names in the namespace) as input and outputs $1$ if $p\in S$, and $0$ if not.  So in this way, any user subscribing to some $p\in P$ will be able to decrypt Alice's message.  See Fig.~\ref{fig:ABE_group} for a visual description.

\paragraph*{Fine-Grained Access Control of ABE.} In this scheme, ABE is used as a means to conveniently interface with the namespace graph.  This alone does not warrant the use of ABE, as other crypto primitives (\emph{e.g.}, hierarchical identity based encryption) can be used for this purpose.  However, the fine-grained access control afforded by ABE allows implementing additional features.  For example, prior work~\cite{SAFE} has implemented an ABE-based ``timeout-based'' revocation procedure, by which compromised or deprecated users can be unsubscribed from a name at the end of a time period.  Our final scheme is compatible with this, and in Section~\ref{sec:eval} we measure the cost of implementing this useful feature.

\begin{figure*}[t]
\begin{minipage}[b]{.45\linewidth}
    \captionsetup{justification=centering}
    \centering
    \includegraphics[width=\linewidth]{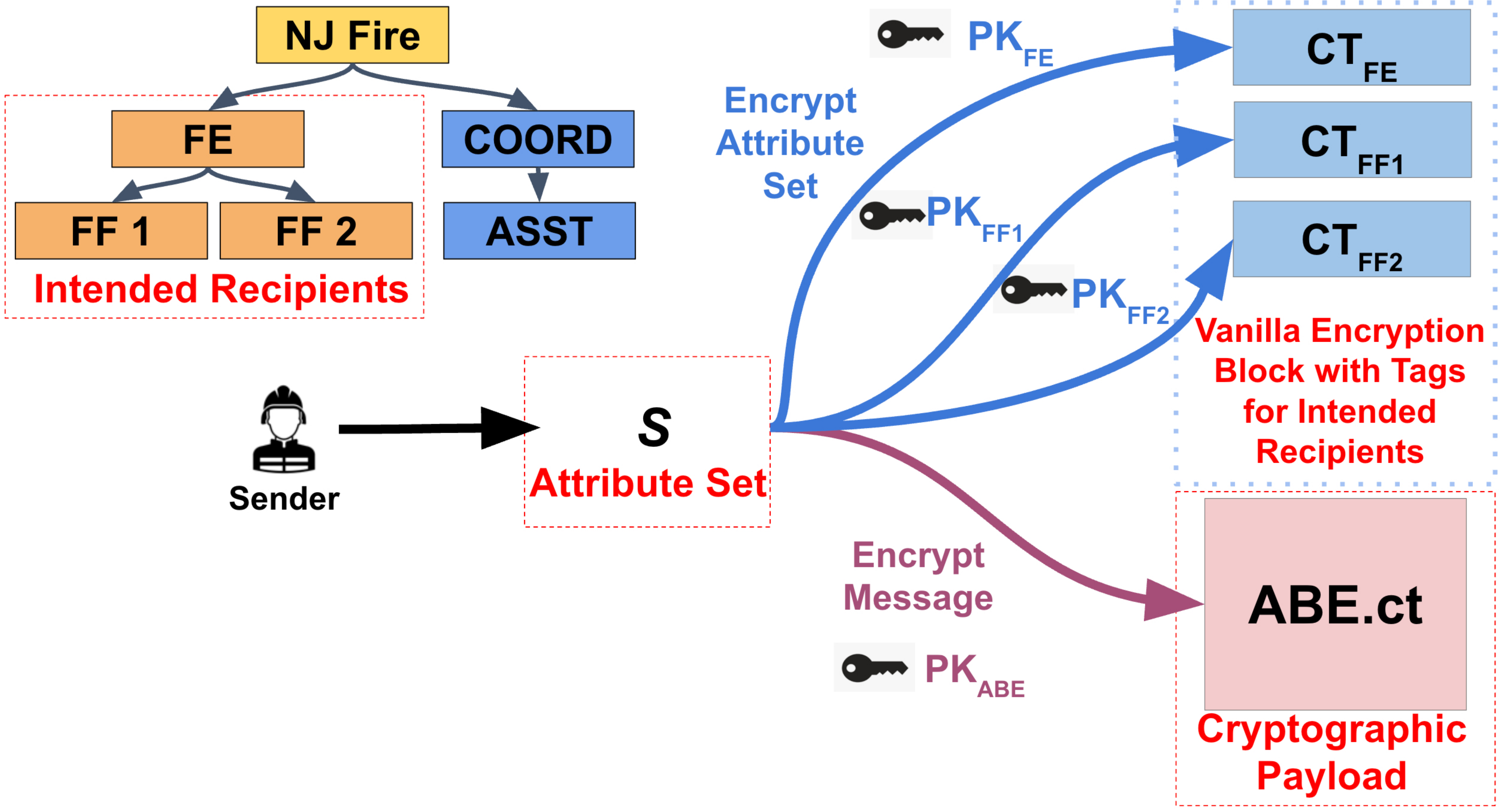}
    \subcaption{{\bf Sender Flow during Encryption}}
    \label{fig:ABEAT_flow_enc}
\end{minipage}\hfill
\begin{minipage}[b]{.5\linewidth}
    \captionsetup{justification=centering}
    \centering
    \includegraphics[width=\linewidth]{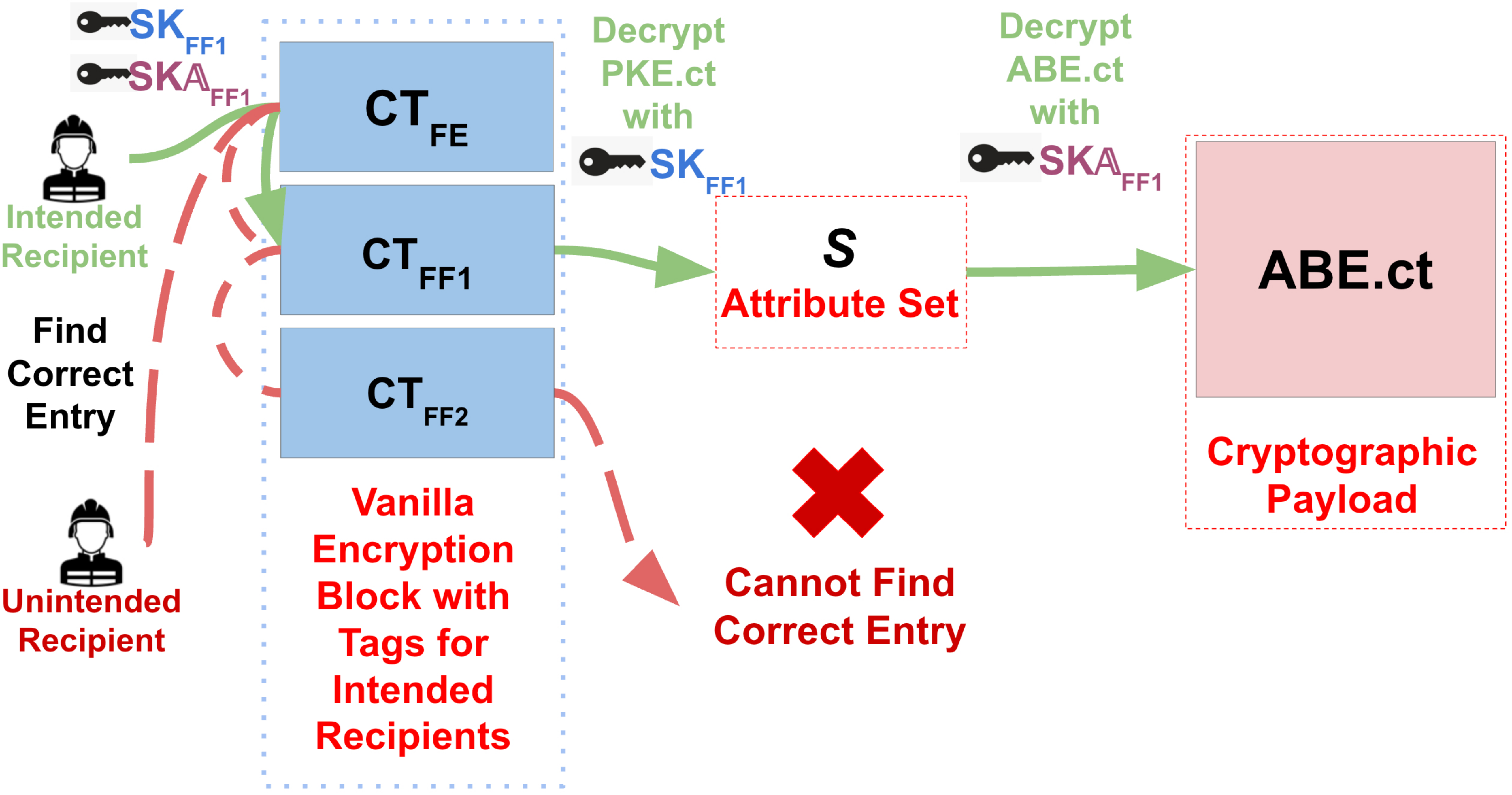}
    \subcaption{{\bf Recipient (both Intended and Unintended) Flow during Decryption}}
    \label{fig:ABEAT_flow_dec}
\end{minipage}
\captionsetup{justification=centering}
\caption{{\bf Example of Encryption and Decryption in \name: Attribute set is encypted in the vanilla encryption block separately from the ciphertext; Tags help the intended recipient to find the correct vanilla encryption block to recover the attribute set}}
\label{fig:ABEAT_flow}
\end{figure*}

\subsection{Adding Recipient Anonymity for ABE-based Group Communication}
We now introduce recipient anonymity for ABE's group communication by encrypting the attributes. Two conditions on the ciphertext pair $(S, \abe.\ct)$ have to be met to be able to provide recipient anonymity:
\begin{itemize}
    \item \textbf{Attribute anonymity:}  
    The attribute set $S$ must remain hidden from the adversary, since the attributes in $S$ are names of the intended recipients when ABE is used for secure group communication.

    \item \textbf{Payload anonymity:}  
    Cryptographic payload $\abe.\ct$ must not leak information enabling an adversary to recover $S$.
\end{itemize}

Payload anonymity has been addressed in previous work like~\cite{FEASE}. In later sections, we will discuss how we extend this scheme to support dual authorities and other desirable features. In this section, we give an overview of our approach for achieving attribute anonymity.
 
\name's key idea %
is to use standard (or \emph{vanilla}) encryption to encrypt each attribute in $S$.  Specifically, we assume every name $p\in\calN$ is associated with some key pair $(\pk_p,\sk_p)$ of a vanilla encryption scheme, where the public key $\pk_p$ is known to all users in the system, while the secret key $\sk_p$ is known only to the users who subscribe to $p$.  With this setup, we modify the previous scheme by having Alice publish $(\mathbb{B}, \abe.\ct)$,
where $\mathbb{B}$ is a collection of vanilla ciphertexts which we call the \emph{vanilla encryption block with tags for intended recipients} (or just \emph{vanilla encryption block}, for short), and $\abe.\ct$ is the cryptographic payload that achieve payload anonymity (confidentiality).  This is shown in Fig.~\ref{fig:ABEAT_flow}.  
The important things about the vanilla block are 1) the intended recipients need to be able to recover the attribute set $S$, to use in decrypting $\abe.\ct$; and 2) other parties (unintended recipients or external adversaries) cannot learn who the intended recipients are.

Note, some nuance is required here. For example, setting $\bbB=\big\{(p,\ct_p)\big\}_{p\in P}$ where for all $p\in P$ (intended recipients), $\ct_p\sim\pke.\enc(\pk_p,S)$, so that each user subscribing to $p\in P$ can use $\sk_p$ to decrypt $\ct_p$ and recover $S$.  Note that while the ciphertexts $\{\ct_p\}_{p\in P}$ are indeed hiding $S$ (hence $P$), any party who sees $\bbB$ learns who the intended recipients are by simply reading them off the indices.  Our solution is to set $\bbB=\big\{(t_p,\ct_p)\big\}_{p\in P}$, where the $t_p$ are cryptographically generated ``tags''.  At a high level, any user subscribed to $p$ can, using the secret key $\sk_p$, compute $t_p$ and find their ciphertext $\ct_p$; while $t_p$ appears random to anyone who does not have $\sk_p$.

\label{sec:properties}

\section{Our Main Contribution: \name}

\subsection{Overview of \name}
Fig~\ref{fig:ABEAT_Process1} shows how the key authorities, sender, and recipients execute the algorithms of \name. 
\begin{itemize}
\item{\textbf{Setup}} This function serves as the global setup for \name and is executed once during system initialization. Its purpose is to generate the public parameters used by subsequent functions, as well as several important global keys: (1) the master public key ($mpk$), which is publicly available and used for encryption; (2) the master secret key ($msk$), which is used to generate decryption keys and must remain confidential; and (3) a public/secret key pair for each name in the namespace ($\{(\pk_p,\sk_p)\}_{p\in\calN}$), which are used in the vanilla encryption block. The public keys are publicly accessible, while the corresponding secret keys are distributed only to subscribers of each name. A key distinction from standard ABE is that \name employs split key authorities to mitigate the risk of master secret key leakage. Specifically, each authority independently generates a share of the master secret key ($msk^{(1)}$ and $msk^{(2)}$). As a result, the setup function is executed jointly by both authorities. They first collaborate to generate the shared parameters (\eg public parameters and the public/secret key pairs for each name), and then each authority generates its respective share of the master secret key (\textcircled{1} in the figure).
\item{\textbf{KeyGen}} This function is used for generating decryption keys for subscribers. Since \name splits key authorities, each key authority (KA) runs this function independently with its part of the master secret key ($msk^{(1)}$ or $msk^{(2)}$) to generate a partial decryption key ($sk_{\mathbb{A}}^{(1)}$ and $sk_{\mathbb{A}}^{(2)}$). Also, the KAs must distribute the secret key of each subscribed name ($\sk_p$) to the corresponding subscriber for use in the vanilla encryption block. All the partial decryption keys and the secret key of each subscribed name are then  transferred to the subscriber, and the subscriber combines them to have the decryption key ($\sk$) (\textcircled{2} in the figure).

\item{\textbf{Enc}} This function is used for encryption when the sender transmits messages. It represents the core component of our design, incorporating the vanilla encryption block. The sender takes the public keys of the target recipient names ($\{\pk_p\big\}_{p\in P}$) as input to construct the vanilla encryption block ($\mathbb{B}$). This block hides the attribute set ($S$), thereby providing attribute anonymity. In addition, the ``tag'' embedded in the vanilla encryption block enables fast recipient verification, allowing unintended recipients to quickly determine that they are not the intended targets, without incurring further decryption cost.
To ensure payload anonymity and message confidentiality, a KP-ABE ciphertext is also generated. Since standard KP-ABE operates under a single authority and provides only message confidentiality (but not payload anonymity), we design a new scheme, \emph{KP-ABE with split key authorities and payload anonymity}. The resulting KP-ABE ciphertext ($\ct_{PL}$) is combined with the vanilla encryption block to form the final ciphertext ($ct$), which is then delivered to the recipients (\textcircled{3} in the figure).
\item{\textbf{Dec}} Upon receiving the ciphertext, each recipient executes \textsf{Dec} to recover the original message. The \textsf{Dec} function first performs recipient verification using the secret key of the subscribed name ($\sk_p$) on the vanilla encryption block to determine whether the recipient is an intended recipient. If not, the process terminates early to avoid unnecessary decryption overhead. Otherwise, the recipient decrypts the vanilla encryption block to obtain the attribute set and then decrypt the KP-ABE ciphertext (\textcircled{4} in the figure).
\end{itemize}

\begin{figure}[t]
\centering
\includegraphics[width=\linewidth]{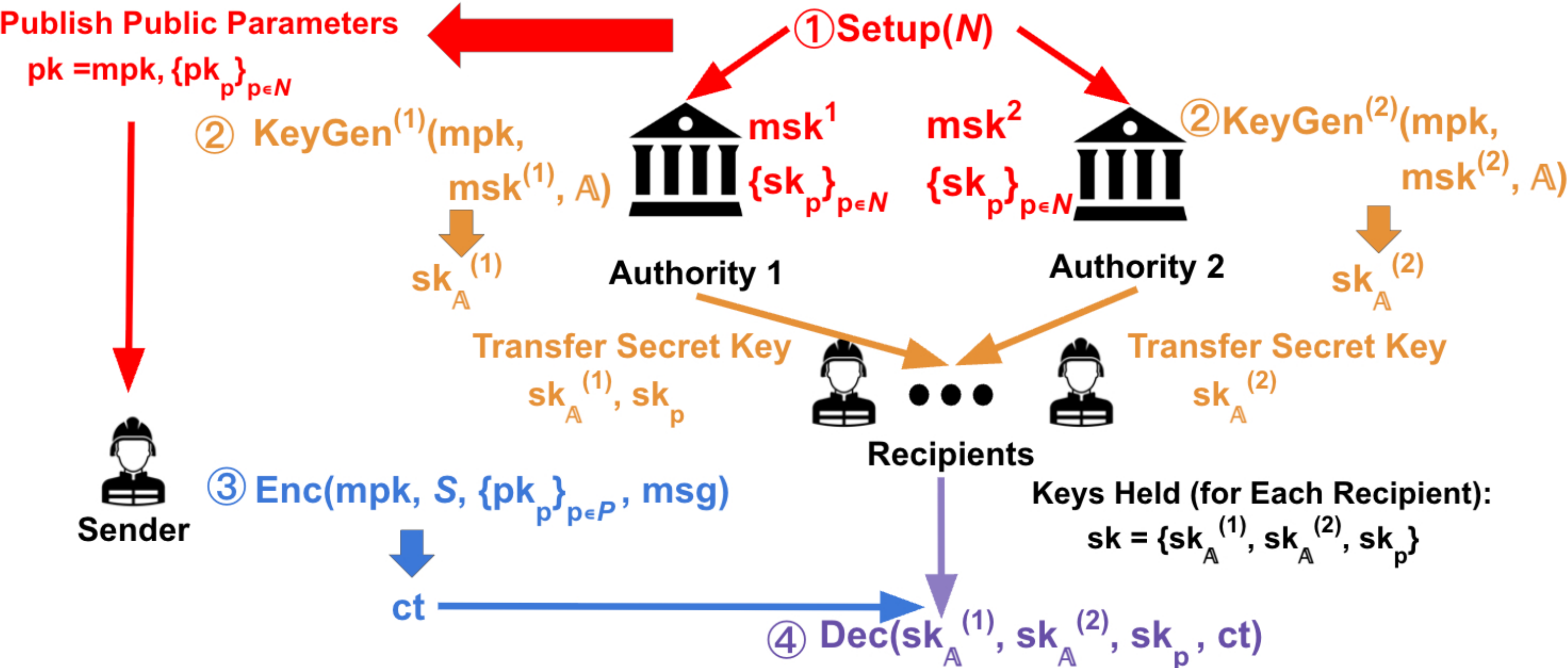}
\caption{\bf Process for \name Algorithms}
\label{fig:ABEAT_Process1}
\end{figure}

As stated before, \name is built on top of two key building blocks: the \emph{KP-ABE scheme with split key-authority and payload anonymity}, and the \emph{vanilla encryption block} scheme. 
First, we describe the building block syntax and how it constructs \name. We then detail the implementation of the building blocks and provide the security proof.

\subsection{Construction of \name}
\label{sec:bulidblocks}
\paragraph*{KP-ABE with Split Key-Authority and Payload Anonymity.} 
We say that the $5$ component algorithms \[(\setup,\keygen^{(1)},\keygen^{(2)},\enc,\dec)\] \emph{have the syntax of a key-policy ABE scheme with the split key authorities (KA)} if they satisfy the input/output of Fig.~\ref{fig:abesyntax2}.

\vspace{-2mm}
\begin{figure}[!h]
    \centering
    \begin{mdframed}[font=\small]
\begin{itemize}
    \item[$\bullet$ \underline{$\abe.\setup(1^\lambda)${\bf :}}] outputs keys $(\mpk,\msk^{(1)},\msk^{(2)})$ 
    \begin{itemize}
        \item[$\cdot$] $\setup$ can also optionally output public information which will be used implicitly by the remaining algorithms.
    \end{itemize}
    \item[$\bullet$ \underline{$\abe.\keygen^{(j)}(\mpk, \msk^{(j)},\mathbb{A})$ for $j=1,2$ {\bf :}}] outputs $\sk^{(j)}_{\mathbb{A}}$.    
    \item[$\bullet$ \underline{$\abe.\enc(\mpk,S,\msg)${\bf :}}] outputs $\ct$.   
    \item[$\bullet$ \underline{$\abe.\dec(\sk_{\mathbb{A}}^{(1)},\sk_{\mathbb{A}}^{(2)},\ct)${\bf :}}] outputs $\msg'$ or the failure symbol $\bot$.
\end{itemize}
    \end{mdframed}
    \vspace{-2mm}
    \caption{Syntax of Key-Policy ABE w/Split KA ($\mathbb{A}$ is access structure)}
    \vspace{-3mm}
    \label{fig:abesyntax2}
\end{figure}

\paragraph*{Block of Vanilla Encryptions.} 
We now define our second building block: a scheme for generating a \emph{vanilla encryption block}. This is the gadget we use for converting an ABE with payload anonymity into an ABE with anonymity for group messaging which supports efficient recipient verification.  
The syntax for the vanilla encryption block scheme is in Fig.~\ref{fig:vanillasyntax2}.  
\vspace{-2mm}
\begin{figure}[!h]
    \centering
    \begin{mdframed}[font=\small]
\begin{itemize}
    \item[$\bullet$ \underline{$\veb.\keygen(\calN)${\bf :}}] outputs a family of key pairs $\big\{(\pk_p,\sk_p)\big\}_{p\in\calN}$.  
    \item[$\bullet$ \underline{$\veb.\enc\bigl(P,\{\pk_p\}_{p\in P},\msg\bigr)${\bf :}}] outputs $\bbB$; here $P\subset\calN$ is a subset.    
    \item[$\bullet$ \underline{$\veb.\dec(\sk_p,\bbB)${\bf :}}] outputs a message $\msg'$, or the failure symbol $\bot$.
\end{itemize}
    \end{mdframed}
    \vspace{-2mm}
    \caption{Vanilla Encryption Block Syntax}
    \vspace{-4mm}
    \label{fig:vanillasyntax2}
\end{figure}

\paragraph*{Construction for \name}
We now give the construction for \name on top of the two key building blocks in Fig.~\ref{fig:our_approach}. 
Specifically, the Setup algorithm invokes $\abe.\setup$ and $\veb.\keygen$. The KeyGen algorithm calls $\abe.\keygen$. The Enc algorithm invokes $\abe.\enc$ and $\veb.\enc$ to produce the ciphertext. 
The Dec algorithm runs $\veb.\dec$ and $\abe.\dec$ to decrypt the ciphertext. 
The detailed construction of the building blocks is described next. 

\begin{figure}[!t]
    \centering
    \begin{mdframed}[font=\small]
    \begin{spacing}{0.1}
    \noindent

$\bullet$ \underline{\bf Building Blocks:} \\
\begin{itemize}
    \item[$\cdot$] a scheme for generating a block of vanilla encryptions \[\veb=\bigl(\veb.\keygen,\veb.\enc,\veb.\dec\bigr);\]
    \item[$\cdot$] a KP-ABE scheme with split key-authority and payload anonymity {\scriptsize\begin{align*}
    \abe=\bigl(\abe.\setup,\{\abe.\keygen^{(j)}\}_{j=1,2},  \abe.\enc,\abe.\dec\bigr).\end{align*}}
\end{itemize}
\vskip 1mm

$\bullet$ \underline{$\setup(\calN)${\bf :}} \\
\begin{itemize}
    \item[$\cdot$] draw $\bigl(\mpk,\msk^{(1)},\msk^{(2)}\bigr)\sim\abe.\setup(1^\lambda)$;
    \item[$\cdot$] draw $\big\{(\pk_p,\sk_p)\big\}_{p\in\calN}\sim\veb.\keygen(\calN)$;
    \item[$\cdot$] output $\bigl(\mpk,\msk^{(1)},\msk^{(2)},\big\{(\pk_p,\sk_p)\big\}_{p\in\calN}\bigr)$.    
\end{itemize}
\vskip 2mm

$\bullet$ \underline{$\keygen^{(j)}(\mpk, \msk^{(j)},\mathbb{A})${\bf :}} for $j=1,2$\\
\begin{itemize}
    \item[$\cdot$] draw and output $\sk_{\mathbb{A}}^{(j)}\sim\abe.\keygen^{(j)}(\mpk, \msk^{(j)},\mathbb{A})$.
\end{itemize}
\vskip 2mm

$\bullet$ \underline{$\enc(\mpk,S,\{\pk_p\}_{p\in P},\msg)${\bf :}} \\
\begin{itemize}
    \item[$\cdot$] draw $(S,\ct_{\sf PL})\sim\abe.\enc(\mpk,S,\msg)$;
    \item[$\cdot$] draw $\bbB\sim\veb.\enc\bigl(\{\pk_p\}_{p\in P},S\bigr)$;
    \item[$\cdot$] output $\ct=(\bbB,\ct_{\sf PL})$.
\end{itemize}
\vskip 2mm

$\bullet$ \underline{$\dec(\sk^{(1)}_{\mathbb{A}},\sk^{(2)}_{\mathbb{A}},\sk_p,\ct)${\bf :}} \\
\begin{itemize}
    \item[$\cdot$] parse $\ct=(\bbB,\ct_{PL})$;
    \item[$\cdot$] compute $S'=\veb.\dec(\sk_p,\bbB)$; if this decryption fails, output the failure symbol $\bot$ and halt;
    \item[$\cdot$] otherwise output $\msg'=\abe.\dec(\sk^{(1)}_{\mathbb{A}},\sk^{(2)}_{\mathbb{A}},S',\ct_{\sf PL})$.  
\end{itemize}
    \end{spacing}
    \end{mdframed}
    \vspace{-3mm}
    \caption{Overall construction of \name}
    \vspace{-2mm}
    \label{fig:our_approach}
\end{figure}

\subsection{Constructing the Building Blocks}

\paragraph*{KP-ABE with Split Key-Authority and Payload Anonymity.} 
We build on top of the KP-ABE construction of~\cite{FABEO}, using ideas similar to those in~\cite{FEASE} for achieving payload anonymity.  Roughly speaking, in~\cite{FEASE}, payload anonymity is achieved by splitting the sender's encryption randomness into two parts, essentially adding an extra degree of freedom which eliminates certain ``anonymity attacks'' which are possible on~\cite{FABEO}.  

Our main idea is to split not just the sender's encryption randomness but the entire key authority.  The result is two independent authorities who each control a share of the master secret key.  In addition to payload anonymity, splitting of the authorities alleviates the risk that the master secret key is leaked (since both shares would now have to be leaked). In the KP-ABE constructions of~\cite{FABEO} and~\cite{FEASE}, a secret value ($\alpha$) is used to provide message confidentiality which must remain confidential. In \name, we require both authorities to collaboratively generate this secret at the beginning of the setup phase. The details of our construction is in Fig.~\ref{fig:abescheme}.

\newcommand{\pprotocol}[5]{{\begin{figure}[#3]
\begin{center}
\fbox{
\hbox{\quad
\begin{minipage}{#4\textwidth}
\small
#5
\end{minipage}
\quad} }
\vspace{-1mm}
\caption{\label{#2} #1}
\vspace{-4mm}
\end{center}
\end{figure} } }

\newcommand{\myprotocol}[4]{\pprotocol{#1}{#2}{h!}{#3}{#4}}

\newcommand{\abeschemeboxcontent}{
\flushleft$\text{ }\bullet\underline{{\bf Public}\text{ }{\bf Parameters:}}$ A pairing group $(\G_1,\G_2,\G_T,g_1,g_2,e)$.

\vskip 2mm\flushleft$\text{ }\bullet\underline{\setup(1^\lambda)}{\bf :}$ A secret $\alpha\sim\mathbb{Z}_p$ is drawn jointly by both key authorities. Specifically, authority $j$ will generate $\alpha_j$, and send $g_1^{\alpha_j}$ to the other party. Both authority will then calculate $g_1^\alpha = \Pi_j\,g_1^{\alpha_j}$. Additionally, two pairs of secrets $\beta_1,\beta_2\sim\mathbb{Z}_p$, $\delta_1,\delta_2\sim\mathbb{Z}_p$ are drawn independently. Both $\beta_j, \delta_j$ are drawn by the key authority $j$.  Output is:
\vskip -1mm
\begin{itemize}
\item[$-$] $\mpk=\bigl(g_2^{\beta_1},g_2^{\beta_2},g_2^{\delta_1},g_2^{\delta_2},g_1^\alpha, e(g_1,g_2)^\alpha\bigr)$;
\item[$-$] $\msk^{(j)}=(\beta_j,\delta_j)$ for $j=1,2$.
\end{itemize}

\vskip 2mm\flushleft$\text{ }\bullet\underline{\keygen^{(j)}(\mpk, \msk^{(j)},\mathbb{A}){\bf :}}$ Parse $\mpk=\bigl(g_2^{\beta_1},g_2^{\beta_2},g_2^{\delta_1},g_2^{\delta_2},g_1^\alpha, e(g_1,g_2)^\alpha\bigr)$ and $\msk^{(j)}=(\beta_j,\delta_j)$ and let $(\bM,\pi)$ be the monotone span program computing $\mathbb{A}$.  Draw $r_j\sim\mathbb{Z}_p$, ${\bf v}_j\sim\mathbb{Z}_p^{n_2}$. 
Output $\sk_{\mathbb{A}}^{(j)}=\bigl(\{\sk_{1,i}^{(j)}\}_{i\in[n_1]},\sk_2^{(j)},\{\sk_{3,i}^{(j)}\}_{i\in[n_1]}\bigr)$:
\begin{itemize}
\item[$-$] $\sk_{1,i}^{(j)}=g_1^{{\bf M}_i\cdot\alpha\cdot{\bf v}_j/\beta_j}\cdot H(\pi(i))^{r_j/\beta_j}$, for $j=1,2$;
\item[$-$] $\sk_2^{(j)}=g_2^{r_j}$, for $j=1,2$.
\item[$-$] $\sk_{3,i}^{(j)}=H(\pi(i))^{r_j/\delta_j}$, for $j=1,2$.
\end{itemize}

\vskip 2mm\flushleft$\text{ }\bullet\underline{\enc(\mpk,S,\msg){\bf :}}$ Parse $\mpk=\bigl(g_2^{\beta_1},g_2^{\beta_2},g_2^{\delta_1},g_2^{\delta_2},e(g_1,g_2)^\alpha\bigr)$.  Draw $s_1,s_2\sim\mathbb{Z}_p$ and let $s=s_1+s_2$.  Output is $(S,\ct_{\sf PL})$, where $\ct_{\sf PL}=\Bigl(\big\{\ct_{1,u}\big\}_{u\in S},\ct_2^{(1)},\ct_2^{(2)},\ct_3^{(1)},\ct_3^{(2)},\ct_4\Bigr)$, where
\begin{itemize}
\item[$-$] $\ct_{1,u}=H(u)^{s}$;
\item[$-$] $\ct_2^{(j)}=g_2^{\beta_{j}s_j}$ for $j=1,2$;
\item[$-$] $\ct_3^{(j)}=g_2^{\delta_js_{3-j}}$ for $j=1,2$;
\item[$-$] $\ct_4=e(g_1,g_2)^{\alpha s}\cdot\msg$.
\end{itemize}

\vskip 2mm\flushleft$\text{ }\bullet\underline{\dec(\sk_{\mathbb{A}}^{(1)},\sk_{\mathbb{A}}^{(2)},S,\ct){\bf :}}$ Parse $\sk_{\mathbb{A}}^{(j)}=\bigl(\{\sk_{1,i}^{(j)}\}_{i\in[n_1]},\sk_2^{(j)},\{\sk_{3,i}^{(j)}\}_{i\in[n_1]}\bigr)$, for $j=1,2$ and $\ct=\bigl(\{\ct_{1,u}\}_{u\in S},\ct_2^{(1)},\ct_2^{(2)},\ct_3^{(1)},\ct_3^{(2)},\ct_4)$.  If $S$ satisfies $(\bM,\pi)$ then find the scalars $\{\gamma_i\}_{i\in I_S}$ such that $\sum_i\gamma_i{\bf M}_i=(1,0,\dots,0)$ and compute: \[{\sf val}:=\prod_{i\in I_S}\prod_{j=1,2}\frac{e\bigl((\sk_{1,i}^{(j)})^{\gamma_i},\ct_2^{(j)}\bigr)e\bigl((\sk_{3,i}^{(j)})^{\gamma_i},\ct_3^{(j)}\bigr)}{e\bigl((\ct_{1,\pi(i)})^{\gamma_i},\sk_2^{(j)}\bigr)}.\]  Output ~$\ct_4/{\sf val}$.
}

\newcommand{\abeschemebox}{
\myprotocol{KP-ABE with Split KA and PA}{fig:abescheme}{.4}
{\abeschemeboxcontent}
}

\abeschemebox
\paragraph*{The Block of Vanilla Encryptions.} 
The Vanilla Encryption Block can be implemented using a standard PKE scheme. Each block is generated for a specific intended recipient and encrypts the corresponding attribute set under that recipient's public key. In addition to the encrypted attribute set, each block also contains a \emph{tag}. This tag enables the intended recipient to efficiently identify which block is meant for them, while not revealing this information to unintended recipients.

Our construction builds on classical Diffie--Hellman key agreement,
where Alice and Bob derive a shared encryption secret. We derive two
shared-secret values, one for tagging and one for message encryption,
as shown in Fig.~\ref{fig:vanillaconstruction}.
\begin{figure}[!h]
\vspace{-2mm}
    \centering
    \begin{mdframed}[font=\small]
\begin{itemize}
    \item[$\bullet$ \underline{$\keygen(\calN)${\bf :}}] For each $p\in\calN$, draw a random exponent $\alpha_p$ and set $(\pk_p,\sk_p)=\bigl(g^{\alpha_p},\alpha_p\bigr)$.
    
    \item[$\bullet$ \underline{$\enc\bigl(P,\{\pk_p\}_{p\in P},\msg\bigr)${\bf :}}] Parse the public keys $\pk_p=g^{\alpha_p}$ for $p\in P$.  Choose random exponents $\beta,\gamma$ and output \[\bbB=\Bigl(g^\beta,g^\gamma,\big\{(g^{\alpha_p\beta},g^{\alpha_p\gamma}\cdot\msg)\big\}_{p\in P}\Bigr).\]
    
    \item[$\bullet$ \underline{$\dec(\sk_p,\bbB)${\bf :}}] Parse $\sk_p=\alpha_p$.  Retrieve $g^\beta$ and $g^\gamma$ from $\bbB$, and compute $g^{\alpha_p\beta}$ and $g^{\alpha_p\gamma}$.  Find in $\bbB$ the pair whose first coordinate is $g^{\alpha_p\beta}$, and let $h$ be the second coordinate of this pair (if there is no such pair, output the failure symbol $\bot$).  Output $h/g^{\alpha_p\gamma}$.
\end{itemize}
\vspace{-1mm}
    \end{mdframed}
    \vspace{-4mm}
    \caption{Generating the Vanilla Encryption Block}
    \vspace{-3mm}
    \label{fig:vanillaconstruction}
\end{figure}

\subsection{Security for \name}

In this section, we discuss the security of \name, which is based on its two main building blocks, the vanilla encryption and  KP-ABE with split key-authority and payload anonymity. We provide a high-level overview of the formal proof, with the Appendix providing the detailed proofs.

\paragraph*{Security of the Vanilla Encryption Block}
Security of our vanilla encryption block component holds assuming the Decisional Diffie–Hellman (DDH) assumption~\cite{DDH}.  Roughly speaking, the DDH assumption states that an adversary cannot distinguish a triple $(g^x,g^y,g^{xy})$ from a uniformly random triple, unless it knows $x$ or $y$.  So, nobody other than the sender (who knows $x$) and the receiver (who knows $y$) can distinguish the tag $g^{xy}$ from a random string.

\paragraph*{Security of KP-ABE with Split Key-Authority and Payload Anonymity}
The security proof of our ABE building block follows the security proof in~\cite{FABEO}, with minor modifications because of our splitting of the key-authority.  Specifically, our proof uses the symbolic security framework of prior work~\cite{FABEO,AC17,ABGW17} within the generic group model (GGM)~\cite{GGM}. Security is argued by expressing the data which is available to the adversary (including public keys, auxiliary secret keys and ciphertexts) as well as the data desired by the adversary (target secret keys) as sets of linear equations, and by showing that the linear span of the equations which the adversary has is disjoint from the linear span of the equations he desires.

\begin{figure}[t]
    \centering
    \includegraphics[width=\linewidth]{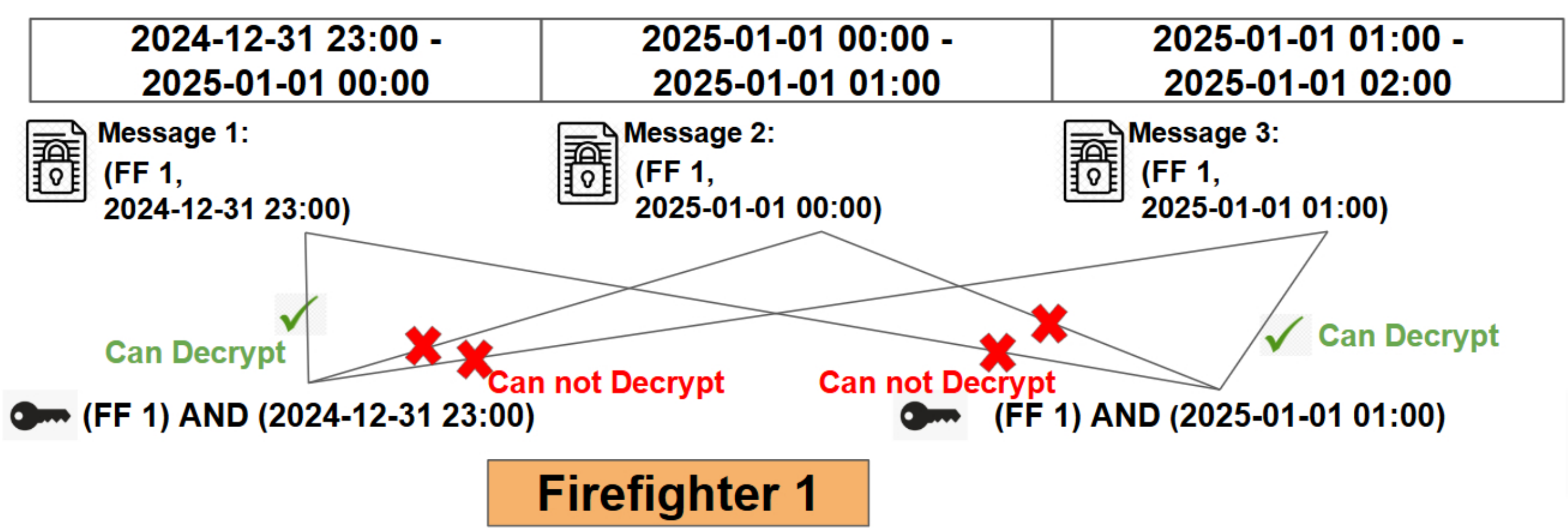}
    \caption{\bf Timeout-Based Revocation, for FF1 as an example}
    \vspace{-2mm}
    \label{fig:access-control}
\end{figure}

\paragraph*{Security of \name} 
The security proof of \name proceeds by reduction to the security of the two building blocks.  Specifically, we show that an adversary who is able to break the security of \name, must also be able to break the security of either the vanilla encryption block or the KP-ABE with split authority and payload anonymity.  This is done using a hybrid argument where a modified version of \name is considered which includes a ``dummy'' vanilla block containing random strings, rather than vanilla encryptions of the intended recipients and tags.  We then show two things: 1) any adversary that can distinguish this hybrid version of \name from the original must be able to directly break the security of the vanilla block; 2) any adversary that can break security of this hybrid version must be able to break security of the KP-ABE with split authority and payload anonymity. 
Since we already proved that the adversary cannot break the security of the vanilla block and the security of the KP-ABE with split authority and payload anonymity, \name is secure.
\paragraph*{Security Properties of \name}

We now summarize the system-level security properties and limitations of
\name:%

\begin{itemize}

    \item 
    \textbf{Identity leakage through the receiver superset.}
     An adversary observing the receiver superset $Q$ may infer the intended recipient set $P \subseteq Q$ through traffic or intersection analysis on repeated transmissions. While \name cryptographically hides $P$ within each $Q$, structural correlation across successive supersets can leak recipient identities. To mitigate this statistical leakage, the sender must construct $Q$ with sufficient size and diversity, ensuring unintended recipients remain uncorrelated with $P$.%

    \item
    \textbf{Leakage of the intended recipient-set size.}
    Because the vanilla encryption block contains one entry for each
    intended recipient name, its length may reveal the number of
    intended recipient names, i.e., $|P|$. This leakage can be
    mitigated by padding the block with randomly generated dummy
    entries that are computationally indistinguishable from valid
    entries to an adversary who does not have the corresponding recipient secret keys. The sender may also use a fixed-size vanilla
    encryption block by adding dummy entries as needed. Consequently, equal-sized blocks prevent an adversary from inferring $|P|$ from the block length.%

    \item 
    \textbf{Forward and backward secrecy.} 
    \name adopts a timeout-based revocation mechanism, as in~\cite{SAFE}.
    As shown in Fig.~\ref{fig:access-control}, time is divided into slots, and each ciphertext and decryption key is bound to a specific slot. Thus, compromising a key exposes only messages from its corresponding slot, but not messages from earlier or later slots, thereby providing backward and forward secrecy respectively. We quantitatively evaluate the cost of this approach in~\ref{sec:access-control}.

    \item %
    \textbf{Confidentiality when an authority is compromised.}
    \name splits the master secret key between two independent authorities, so compromising either authority alone does not allow payload decryption. The authorities must collude to allow for recovery of the protected messages, while recipient anonymity is provided separately by the vanilla encryption block and payload-anonymous KP-ABE.

\end{itemize}

\begin{figure*}[t]
\begin{minipage}[b]{.32\linewidth}
    \captionsetup{justification=centering}
    \centering
    \includegraphics[width=\linewidth]{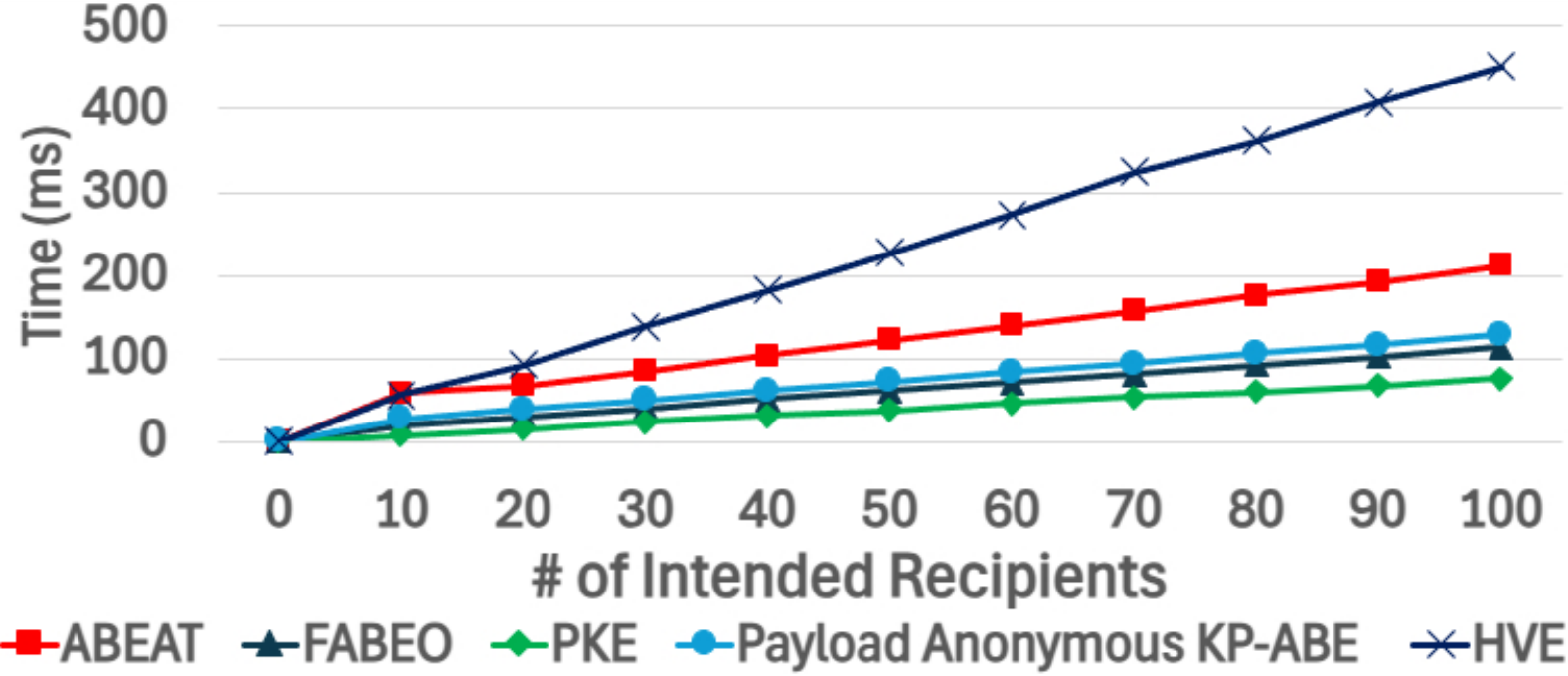}
    \captionof{figure}{{\bf Encryption Time Comparison }}
    \label{fig:enc_total}
\end{minipage}
\begin{minipage}[b]{.32\linewidth}
    \centering
    \includegraphics[width=\linewidth]{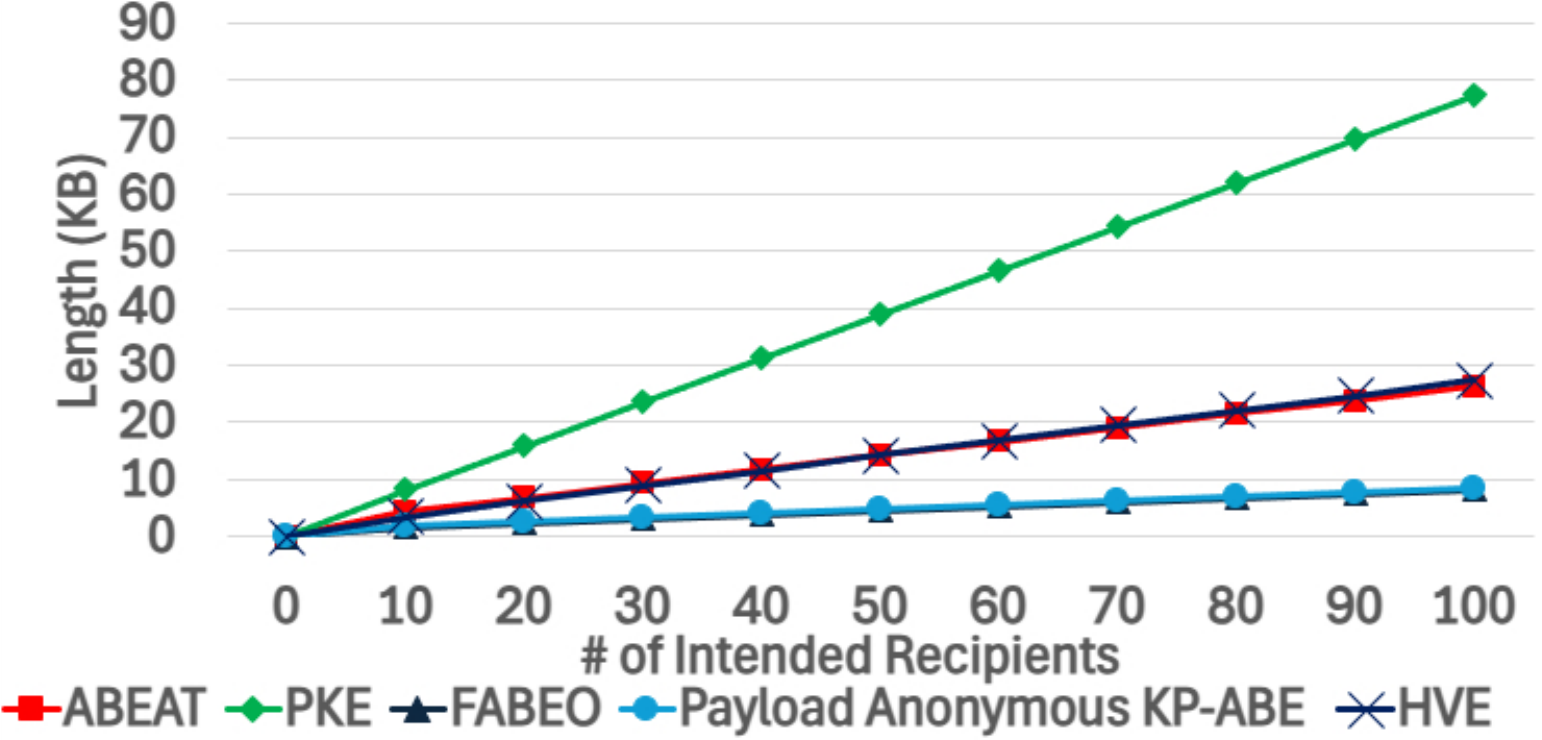}
    \caption{{\bf Ciphertext Length Comparison}}
    \label{fig:enc_cipher_len}
\end{minipage}
\begin{minipage}[b]{.32\linewidth}
    \centering
    \includegraphics[width=\linewidth]{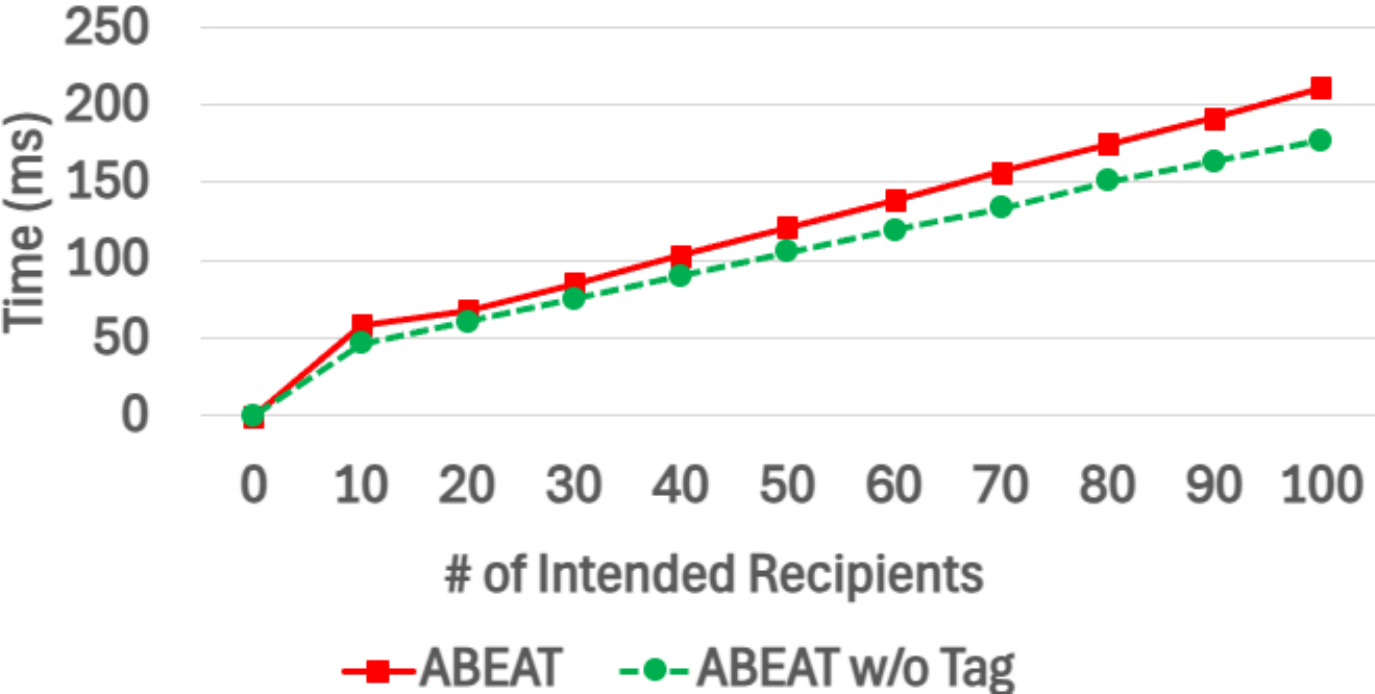}
    \caption{{\bf \name Encryption: w \& w/o Tag}}
    \label{fig:enc_no_tag}

\end{minipage}
\end{figure*}

\section{Evaluation}
\label{sec:eval}

\subsection{Implementation details}
\label{sec:implementation}

We implement \name using Charm~\cite{charm}, with MNT224 pairing groups~\cite{MNT224} and prime192v1 elliptic-curve
PKE~\cite{prime192v1}. For key encapsulation, a random group element is used to derive a symmetric payload-encryption key. All schemes are
evaluated on the same 2.1\,GHz server, with results averaged over five
runs.

\subsection{Cost of Providing Anonymity with ABE} 
We first understand the cost of supporting anonymity in \name compared to alternative approaches, including (1) HVE~\cite{hve08}; (2) KP-ABE with Split Key-Authority and Payload Anonymity. (as described in Section~\ref{sec:bulidblocks}, we call it Payload Anonymous KP-ABE for short), but with the attribute set in the clear; (3) FABEO~\cite{FABEO}; and (4) A simple concatenation of $\pke$ encryption for all intended recipients. 
Both $\pke$ and FABEO do not provide anonymity. $\pke$ also does not provide fine-grained access control. Payload Anonymous KP-ABE provides only weaker anonymity as it does not hide the attribute set, which is in the clear. HVE provides full anonymity, but the encryption overhead depends on the size of the entire namespace (all possible receivers of a message). Note that the performance of all schemes except HVE is independent of namespace size. To conservatively estimate the relative encryption cost benefit of \name, we set the HVE namespace size to be only twice the intended-recipient set.

\paragraph*{Encryption} 

Fig.~\ref{fig:enc_total} shows the encryption cost for \name. $\pke$, without anonymity is certainly the fastest, with FABEO being slightly slower than that.  
Payload Anonymous KP-ABE is also slightly faster than \name, but it 
also does not provide full anonymity (it does not hide the attribute set).
\name provides a higher anonymity level, and takes more time for encryption. The proportionate increase in latency for encryption with \name compared to the other non-anonymous approaches reduces slightly as we increase the number of intended recipients. HVE is by far the slowest.
Even when we limit the namespace to only be twice as large as the intended recipient set, HVE's overhead grows very rapidly compared to \name.  
Even with a namespace that contains just 2000 names, HVE would take around 4480 ms for encryption (with 100 intended recipients), while \name would take 211.47ms since the cost of \name does not depend on namespace size. Overall, adding anonymity with \name incurs about 88\% more overhead compared to FABEO.

\paragraph*{Decryption}

Table~\ref{tab:dec_total} compares the decryption cost for intended and unintended recipients across different approaches. For intended recipients, \name incurs about $2\times$ the decryption cost of Payload Anonymous KP-ABE. FABEO, which provides neither anonymity nor dual authorities, is faster, while $\pke$, also without anonymity, remains the fastest. HVE provides anonymity but not dual authorities, and has the highest cost.

For unintended recipients, FABEO and Payload Anonymous KP-ABE provide
no or weaker anonymity and take about $0.8$\,ms to reject a ciphertext.In contrast, $\pke$ requires trying every recipient ciphertext, taking $38.8$\,ms for $100$ intended recipients, while HVE gives intended and unintended recipients comparable costs.   
\name makes unintended-recipient
decryption about $80\times$ faster than intended-recipient decryption,
over $90\times$ faster than HVE, and $40\%$ faster than non-anonymous
FABEO, approaching the cost of non-anonymous schemes.

\paragraph*{Ciphertext Length}
Fig.~\ref{fig:enc_cipher_len} shows the total ciphertext length resulting from each approach. We use the serialization method provided in Charm to convert the ciphertext to bytes, to estimate length. Non-anonymous $\pke$ group communication has the largest ciphertext length, because the message payload has to be replicated for each recipient when using $\pke$ to send to multiple recipients. HVE and \name have similar lengths, while the ciphertext length of FABEO and Payload Anonymous KP-ABE~\cite{FEASE} is smaller. 

\begin{table}[t]
    \centering
    \captionof{table}{ {\bf Decryption Time (ms) \& Encryption CPU Cycles ($10^6$) }}
    \footnotesize
    \begin{tabular}{c|c|c|c|c|c}
        \hline
        \multirow{2}{*}{\textbf{Scheme}}&\multicolumn{2}{|c}{recipient}& \multicolumn{3}{|c}{encrypt cycles (\#intended)}\\
        \cline{2-6}
        &intended&unintended&20&50&100\\
        \hline
        \name &$35.32$&$0.42$&$140.91$&$253.31$&$439.46$ \\
        \hline
        HVE &$40.69$&$40.67$&$193.58$&$477.29$&$946.83$ \\
        \hline
        P-A-KP-ABE & $17.77$ & $0.82$&$78.15$&$143.93$&$254.31$\\
        \hline
        FABEO & $12.97$ & $0.84$&$62.36$&$127.67$&$238.16$\\
        \hline
       \multirow{2}{*}{PKE}& \multirow{2}{*}{$0.43$} & $38.8$&\multirow{2}{*}{$32.19$}&\multirow{2}{*}{$78.99$}&\multirow{2}{*}{$157.6$}\\
        &  & 100recipients&&&\\
        \hline
    \end{tabular}
    \label{tab:dec_total}
    \vspace{-2mm}
\end{table}

\paragraph*{CPU cycles}
The utilization of the encryption CPU is similar to the encryption processing time as shown in Table~\ref{tab:dec_total}.
HVE takes the most cycles during encryption. It takes around 946M cycles for 100 recipients. \name is the next. For 100 recipients, the encryption process requires approximately 410M CPU cycles while the payload anonymous KP-ABE (P-A-KP-ABE) requires 254M cycles. FABEO KP-ABE is next, requiring 238M cycles. In contrast, PKE group communication requires the fewest CPU cycles, with 157M cycles for 100 recipients.

\subsection{Performance of \name with Tags}
The Diffie-Hellman based tags could dramatically reduce the decryption processing time at recipients. However, \name adds encryption processing at the sender. To quantitatively measure these overheads of \name, we compare it with an approach that only contains the PKE ciphertext in the vanilla encryption block. Note that the decryption savings apply to (potentially) a very large number of unintended recipients, while the encryption overhead is only for a single sender. 

\paragraph*{Encryption Overhead:} 
Fig.~\ref{fig:enc_no_tag} shows the encryption time for \name with and without the use of the tag. For adding a tag, \name incurs a slightly higher encryption time. The overhead added for encryption varies from 10-30ms based on the number of intended recipients. The overhead added for 20 recipients is 11\%, and is about 20\% for 100 recipients.  

\paragraph*{Decryption Savings:} 
Table~\ref{tab:dec_RV} shows the cost of decryption for \name with and without the tag. 
For an intended recipient with a tag, the decryption cost has three components: recipient verification and $\pke$ decryption, which together take $0.77$\,ms; ABE payload decryption, which takes $34.31$\,ms; and the search time for locating the correct tag, denoted as $ST$ in the table. The search time is typically negligible. For large vanilla encryption blocks, we use binary search to reduce it, as shown in Appendix~\ref{sec:BS}. For an unintended recipient with a tag, only recipient verification, which is comparable to a single $\pke$ decryption, is required.
Without tags, both intended and unintended recipients incur substantially higher overhead. Suppose the $\pke$ ciphertext is stored in the $t$-th entry. An intended recipient must perform $t$ $\pke$ decryptions, each taking about $0.42$\,ms, before running a full KP-ABE decryption, which takes $17.6$\,ms, to recover the message. For an unintended recipient, the cost depends on the vanilla encryption block size $M$. Since none of the $\pke$-encrypted attribute sets can be decrypted successfully, the recipient must try all entries, resulting in a cost of $M$ $\pke$ decryptions.

\begin{figure*}[t]
\centering

\begin{minipage}[t]{0.34\linewidth}
\vspace{0pt}
\centering
\includegraphics[width=0.9\linewidth]{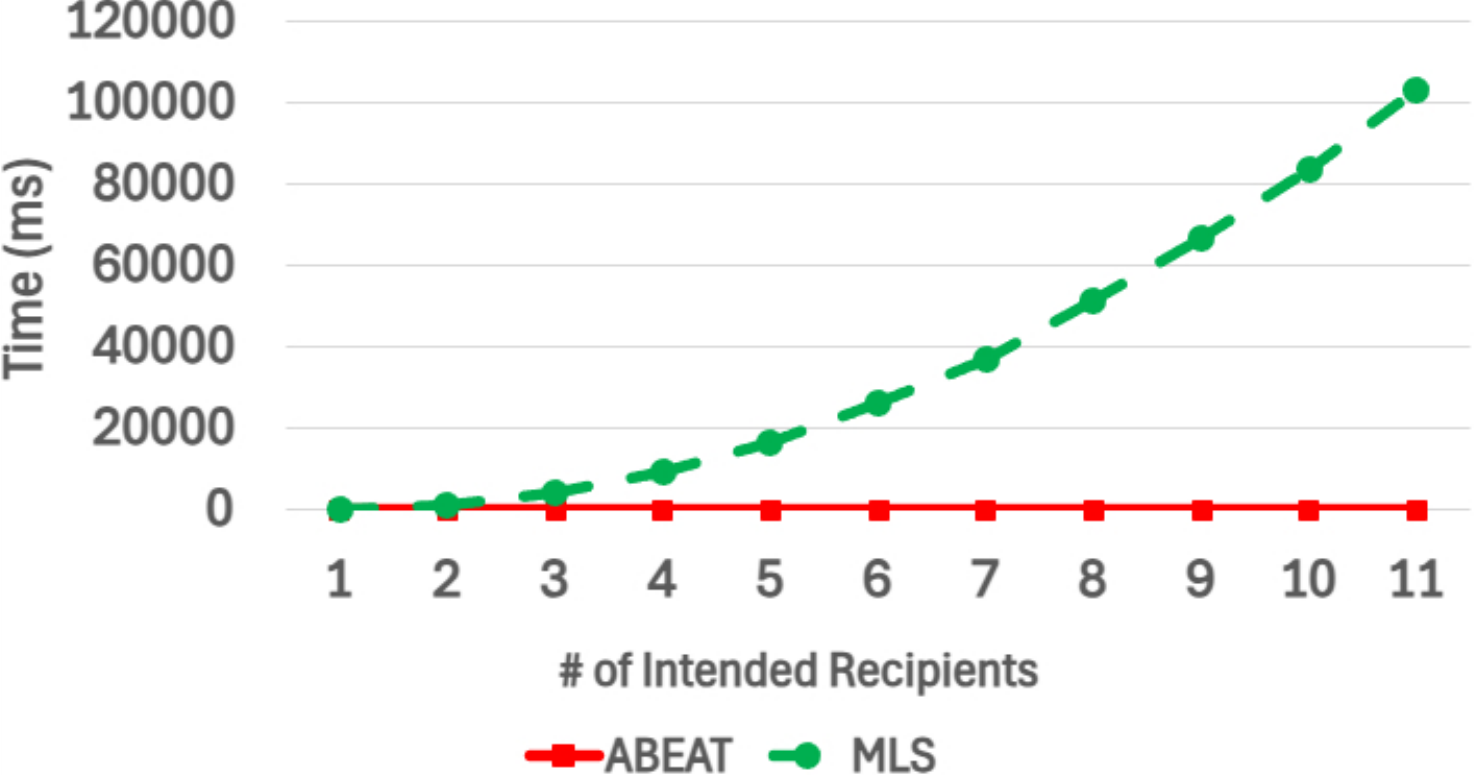}
\vspace{-2mm}
\captionof{figure}{{\bf \name\!\vs\!MLS: Group Setup \& Enc}}
\label{fig:mls}
\end{minipage}
\hfill
\begin{minipage}[t]{0.62\linewidth}
\vspace{-2mm}
\centering

\begin{subfigure}[t]{0.48\linewidth}
    \centering
    \includegraphics[width=\linewidth]{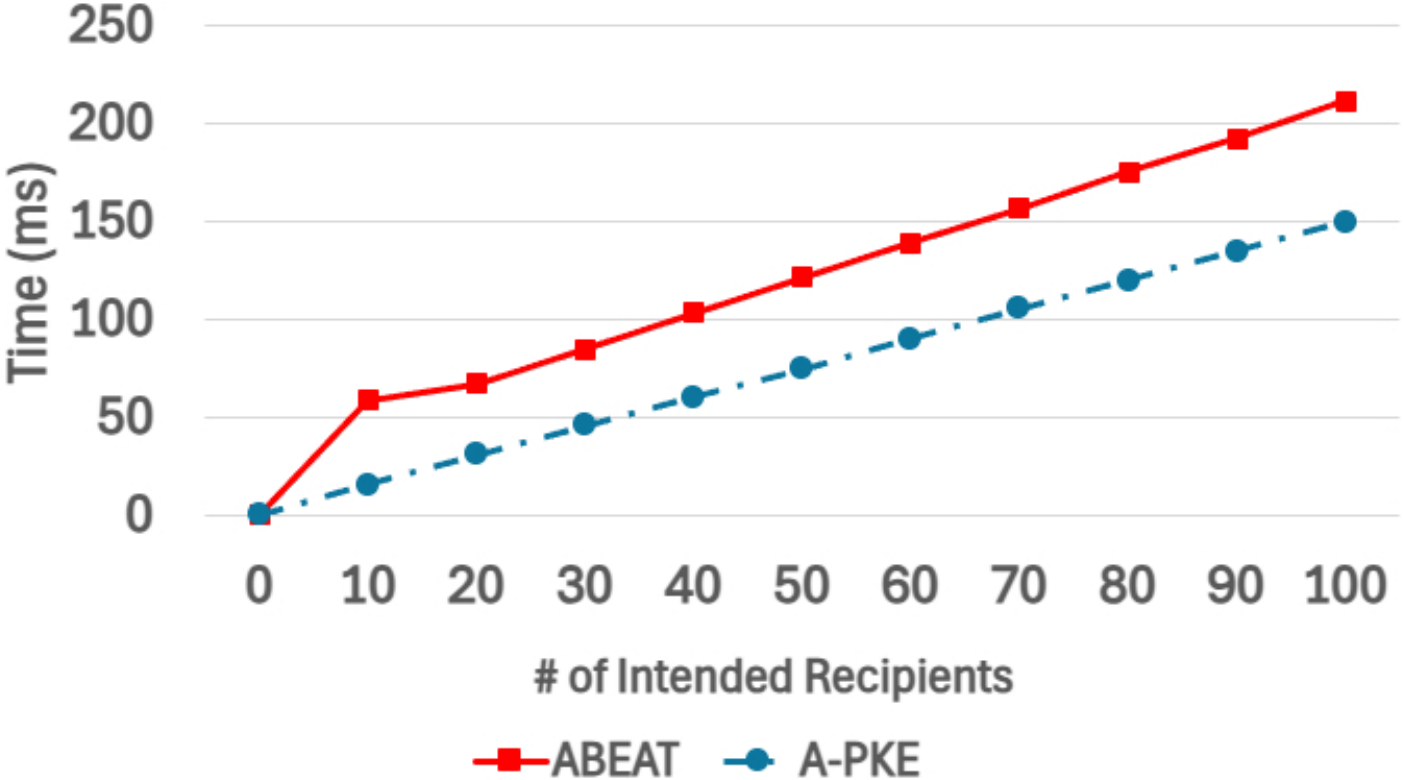}
    \caption{{\bf Encryption Time.}}
    \label{fig:PKE_enc}
\end{subfigure}
\hfill
\begin{subfigure}[t]{0.48\linewidth}
    \centering
    \includegraphics[width=\linewidth]{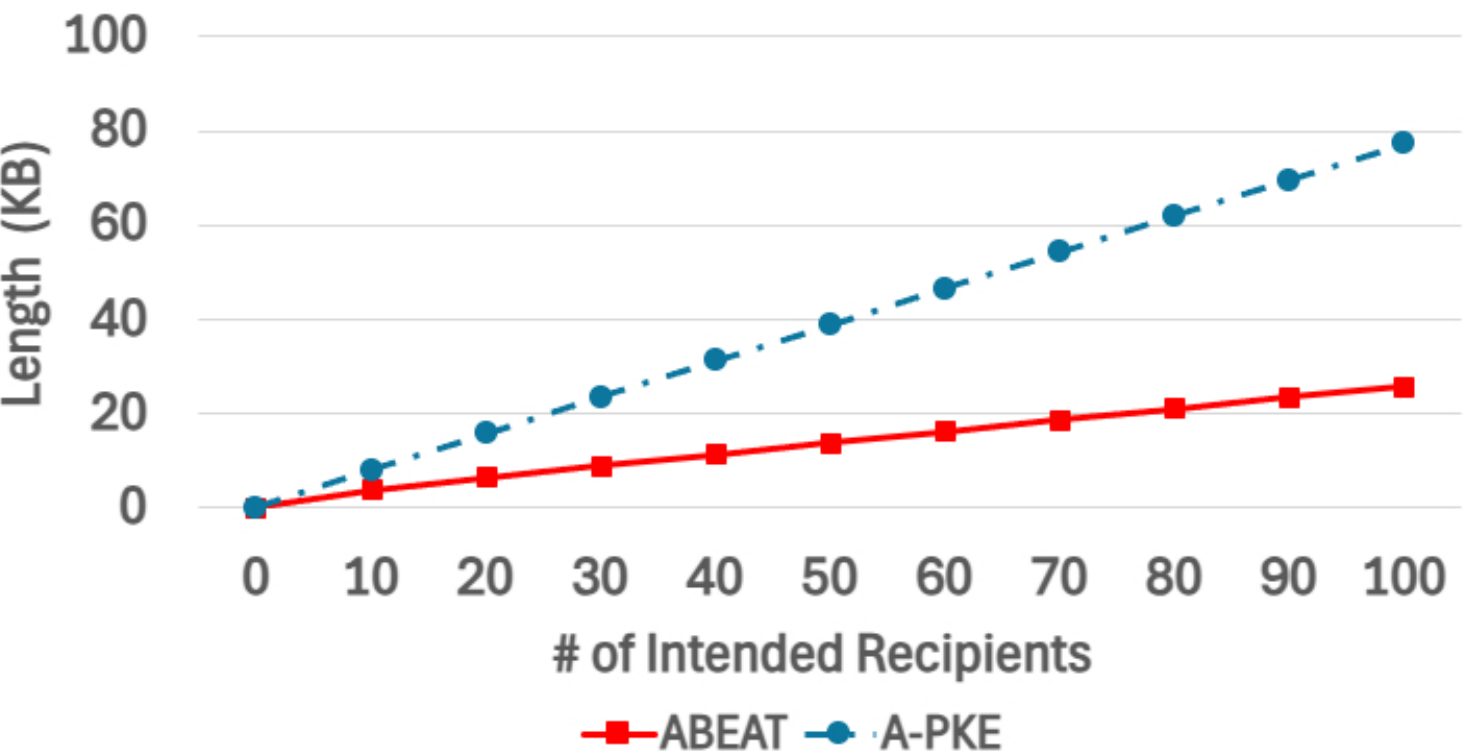}
    \caption{{\bf Ciphertext Length.}}
    \label{fig:PKE_len}
\end{subfigure}

\vspace{-2mm}
\caption{{\bf Comparison between \name and PKE}}
\label{fig:ABEAT_APKE}
\end{minipage}
\end{figure*}

\subsection{\name w/  Access Control and Revocation}
\label{sec:access-control}
We implement the timeout-based revocation in~\cite{SAFE} by adding one timestamp attribute. 
Our evaluation result shows it adds negligible overhead to encryption, as seen in Table~\ref{tab:access_control}, even supporting 100 intended recipients with the timestamp.

The decryption cost for adding fine-grained access control is shown in Table~\ref{tab:access_control}. Adding a single fine-grained access control attribute adds just around 1\% overhead to \name. 

\begin{table}[h]
\centering
    \vspace{-2mm}
    \captionsetup{justification=centering}
    \caption{{\bf \name Decryption: w \& w/o Tag ($ST$: searching time; $t$: \# of probes; $M$: size of vanilla encryption block) (ms)}} 
    \vspace{-2mm}
    \begin{tabular}{c|c|c}
        \hline
        \textbf{Scheme}&intended recipient& unintended recipient\\
        \hline
        \name w/ Tag&$0.77 + 34.31 + ST$&$0.42$ \\
        \hline
        \name w/t Tag & $t \times 0.42 + 34.31$&$M \times 0.42$\\
        \hline
    \end{tabular}
    \label{tab:dec_RV}
    \vspace{-1mm}
\end{table}

\begin{table}[h]
\vspace{-2mm}
\captionof{table}{ {\bf \name w/ Fine-grained Access Control (ms) }}
\vspace{-2mm}
    \footnotesize
    \begin{tabular}{c|c|c|c|c|c}
        \hline
        \multirow{2}{*}{\textbf{Scheme}}&\multicolumn{3}{|c|}{encryption (\# intended rec.)}& \multicolumn{2}{|c}{ recipient}\\
        \cline{2-6}
        &20&50&100&intended&unintended\\
        \hline
        \name &$67.17$&$120.37$&$209.94$&$35.17$&$0.42$ \\
        \hline
        \name &\multirow{2}{*}{$68.31$}&\multirow{2}{*}{$121.81$}&\multirow{2}{*}{$210.71$}&\multirow{2}{*}{$35.52$} &\multirow{2}{*}{$0.39$}\\
        w/ Access Control &&&&&\\
        
        \hline
    \end{tabular}
    \label{tab:access_control}
    \vspace{-2mm}
\end{table}

\subsection{Key Generation Cost in \name}
We now analyze the key generation time in \name. Key generation is involved in three operations: namespace initialization, adding a new name, and subscribing to a name. Table~\ref{tab:keygen_cost} reports the cost of each operation. During initialization, the two KAs generate their master public/private key pairs, taking $\sim 46.50$\,ms in total. In addition, each initial name in the namespace requires one $\pke$ key pair, which takes $0.33$\,ms to generate. Thus, for $N$ initial names, the total namespace initialization cost is $N \times 0.33 + 46.50$\,ms. Adding a new name only requires generating one $\pke$ key pair, taking $0.33$\,ms. When a subscriber subscribes to a name, both authorities generate an ABE decryption key. For one attribute, generating the two partial ABE decryption keys takes $24.35$\,ms in total.

\begin{table}[h]
\centering
\vspace{-2mm}
\captionsetup{justification=centering}
     \captionof{table}{{\bf Cost (time) for Key Generation (ms)}}
     \vspace{-2mm}
    \footnotesize
    \begin{tabular}{c|c|c|c}
        \hline
        \textbf{Scheme}&initialization & adding a name & subscribing\\
        \hline
        \name&$N \times 0.33 + 46.50$&$0.33$&$24.35$ \\
        \hline
    \end{tabular}
    \label{tab:keygen_cost}
    \vspace{-5mm}
\end{table}

\subsection{\name vs. PKE: Is ABE really necessary?}

A common concern with ABE is its overhead. However, using only PKE introduces significant key management complexity. In a PKE-only design, all members of a relatively static group must agree on a common group key pair, which also raises the traditional problem of group key revocation. In contrast, ABE allows each \name message to define a distinct, dynamically formed group based on its attributes. Although \name still requires each recipient to hold a PKE key pair, this key pair is managed through the namespace. Recipients obtain it when subscribing to a name, avoiding complex key management.

We also quantify the cost of anonymity. Fig.~\ref{fig:ABEAT_APKE} compares \name with anonymous PKE (A-PKE), where A-PKE uses the same vanilla encryption block as \name. The encryption time of \name is about $1.4\times$ that of A-PKE for $100$ recipients. However, as shown in Fig.~\ref{fig:PKE_len}, the ciphertext length of \name is about $66\%$ shorter than that of A-PKE. Given the key management complexity and longer ciphertexts of A-PKE, we argue that building \name on ABE is a better choice for anonymous group communication than using PKE alone.

\subsection{Comparison with Message Level Security}
MLS~\cite{mls} targets commercial group messaging over stable networks
and does not provide recipient anonymity.  We compare \name with an open-source implementation of MLS,
OpenMLS~\cite{openmls}, under a highly dynamic setting in which group
membership changes for every message. MLS therefore performs group
formation before each transmission, whereas \name requires no explicit
group formation. In our experiment, all group members are hosted on the same machine, to eliminate network transmission delays for control messages during group formation. This setup obviously favors MLS. We measure the combined group-formation and encryption time in
Fig.~\ref{fig:mls}. Even without network delays, OpenMLS is over
$200\times$ slower than \name for 50 members. System load remained well below CPU capacity, confirming that the difference results from
protocol overhead rather than system overload.

\subsection{Large Scale Evaluation}
\begin{table}[t]
    \captionof{table}{ {\bf Total Cost for Large Scale Evaluation (ms) }}
    \footnotesize
     \begin{tabular}{c|c|c|c|c}
        \hline
        \textbf{Setting}&\textbf{Scheme}&sender&intended recipients & unintended recipients\\
        &&&(total cost)& (total cost) \\
        \hline
        $2\times$&\name&$566.41$&$8382.85$&$371.80$ \\
        \cline{2-5}
        Intended &Baseline&$275.65$&$4211.68$&$8418.09$ \\
        \hline
        \multirow{2}{*}{Broadcast}&\name&$564.23$&$8370.85$&$899.04$ \\
        \cline{2-5}
        &Baseline&$272.97$&$4210.99$&$19649.00$ \\
        \hline
    \end{tabular}
    \label{tab:large_unintended}
    \vspace{-2mm}
\end{table}

We generate a random namespace graph with 1337 nodes and select a random node as the recipient name. From this node, we expand toward the leaves to form a recipient group, resulting in 236 intended recipients. For unintended recipients, we randomly select 472 nodes (twice the number of intended recipients) from the graph.
We compare \name to a baseline anonymous ABE approach that is a combination of a mechanism using bloom filters as described in~\cite{HideCP19} and the payload anonymous KP-ABE we described previously~\cite{FEASE}. Table~\ref{tab:large_unintended}   the overall results. Although the total cost for the sender is higher with \name, the cost for all 
unintended recipients (summed together) is dramatically lower, by about a factor of 20x. The total computation for all intended recipients together is around 2x the bloom filter/payload anonymous KP-ABE baseline. However, \name provides the dual KA feature that the baseline approach does not have.
Thus, when we consider the overall system overhead, \name is much lower.

We also test the same large scale experiment using 
broadcast in the underlying network layer. Here, all nodes in the network that are not intended recipients will be considered unintended recipients. Since set size increases for broadcast, the cost for the baseline approach increases considerably, whereas 
\name still maintains the low overhead for unintended recipients.

\subsection{Overall Solution Evaluation}
This section evaluates the end-to-end performance of \name when integrated into the name-based group messaging framework. Our experiments are conducted using the network simulator provided by POISE~\cite{networksimulator}, with simulation settings comparable to those adopted in its original evaluation. The simulated environment includes multiple forwarding routers and end hosts serving as publishers and subscribers. For every published message, we extend the simulation to account for the corresponding encryption and decryption operations.

\begin{figure}[t]
    \centering
    \includegraphics[width=\linewidth]{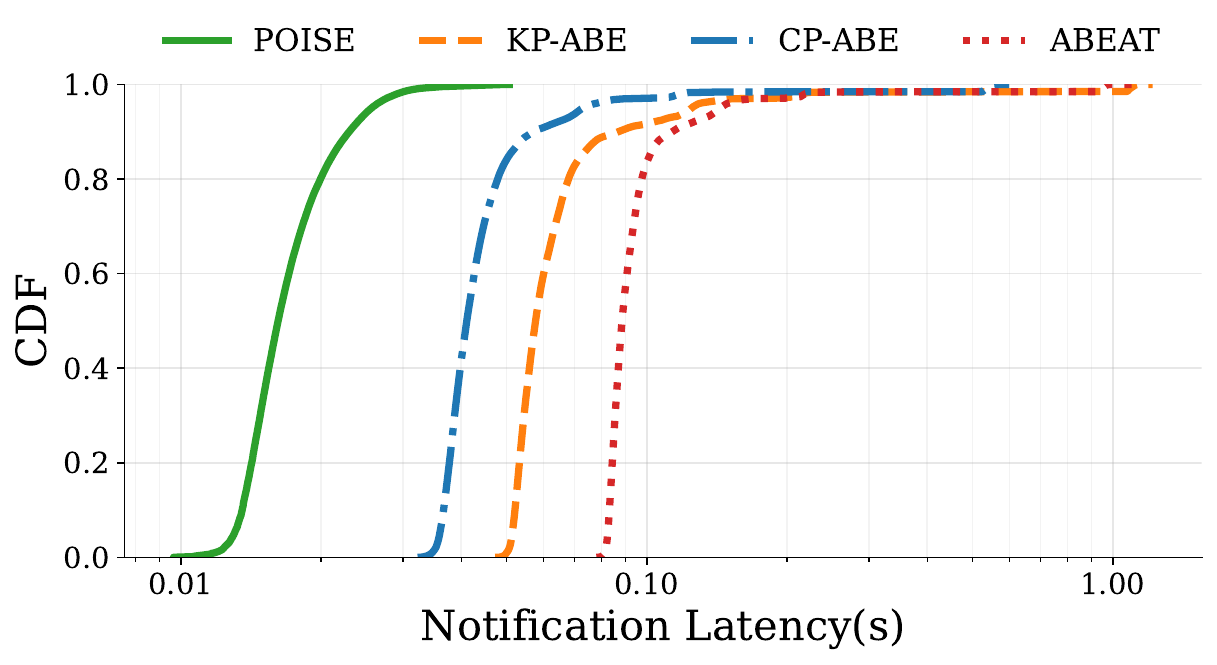}
    \vspace{-2mm}
    \caption{Overall Solution Evaluation}
    \label{fig:overall}
    \vspace{-2mm}
\end{figure}

The simulated network is based on the Rocketfuel 1221 Telstra topology~\cite{Rocketfuel}. For the naming structure, we construct a graph-based ``Disaster Management'' category namespace from the Wikipedia dataset~\cite{nameWiki}. The resulting hierarchy comprises 1,468 names distributed across six levels. Using the corresponding Wikipedia pages and files, we create a workload of 95,417 publications and shuffle their transmission order before running the simulation.

Our primary metric is notification latency, defined as the total time spent delivering a publication and performing its associated encryption and decryption operations. Fig.~\ref{fig:overall} shows that incorporating these cryptographic operations increases the end-to-end notification latency. Without confidentiality and anonymity protection, POISE delivers 90\% of publications within 0.02,s. At the same percentile, the KP-ABE and CP-ABE solutions require 0.05,s and 0.08,s, respectively, while \name requires 0.1,s. Although the anonymity mechanisms introduce noticeable overhead relative to the network delay alone, the resulting latency remains acceptable for disseminating messages to first responders.

\vspace{-1mm}
\section{Conclusion}
\vspace{-1mm}
In this paper, we proposed \name, an ABE-based system for confidential dynamic group communication in practical settings such as emergency response. \name enables flexible group management by incorporating the namespace graph into the ABE setup. It also provides two key security guarantees through new cryptographic modules. First, \name achieves recipient anonymity, preventing adversaries who monitor network traffic from identifying the intended recipients, by adding a block of standard vanilla encryptions to the ABE ciphertext. These vanilla ciphertexts are tagged to support efficient recipient verification, allowing users to cheaply check whether a message is intended for them before performing costly ABE decryption. Second, \name reduces the risk of master-secret-key leakage by splitting the ABE key authority into two dual authorities. Our evaluation showed that \name reduces computation time for unintended recipients by more than $90\times$ compared to an anonymous scheme such as HVE, and by $40\%$ compared to even the non-anonymous FABEO scheme.

\section{Acknowledgement}
\label{sec:ack}
\vspace{-1mm}
This work was supported by NSF grants CNS-2434009 and CAREER-2441313.

\bibliographystyle{plain}

\bibliography{bib/references}

\clearpage
\appendices 

\section{Security Proof}
\subsection{Security of the Vanilla Encryption Block}
\label{app:vanilla}
In this section we describe the corrupt or target DDH assumption and use it to prove security of our vanilla encryption block scheme.

We first define the security game as showed in Fig.~\ref{fig:vanillasecurity}.
\begin{figure}[!h]
    \centering
    \begin{mdframed}[font=\small]
\begin{itemize}
    \item[{\bf 1. Setup:}] $\calC$ draws $\big\{(\pk_p,\sk_p)\big\}_{p\in\calN}\sim\keygen(\calN)$ and sends all public keys to $\calA$: $\{\pk_p\}_{p\in\calN}$.  Additionally, $\calC$ chooses a random bit $b\sim\{0,1\}$.

    \item[{\bf 2. Queries:}] The game now enters the query phase and remains here until $\calA$ chooses to proceed.  During each round of the query phase, $\calA$ will issue one of two types of queries to $\calC$ and $\calC$ will respond.
    \begin{itemize}
        \item[(i) \underline{Secret Key Queries} $-$] $\calA$ sends $p\in\calN$ to $\calC$; $\calC$ sends $\sk_p$ to $\calA$.
        \item[(ii) \underline{Ciphertext Queries} $-$] $\calA$ sends $(P,\msg)$ to $\calC$, where $P\subset\calN$; $\calC$ returns $\bbB$ to $\calA$ where $\bbB$ is prepared as follows:
        \begin{itemize}
            \item[$\cdot$] if $b=0$: $\bbB\sim\enc\bigl(P,\{\pk_p\}_{p\in P},\msg\bigr)$;
            \item[$\cdot$] if $b=1$: $\bbB\sim\$$ is drawn uniformly at random.
        \end{itemize}        
    \end{itemize}

    \item[{\bf 3. Guess:}] $\calA$ sends $b'\in\{0,1\}$ to $\calC$, and the game ends.
    
    \item[{\bf Outcome:}] $\calA$ wins if b'=b and additionally if:
    \begin{itemize}
        \item[$-$] $p\notin P$ for all $p\in\calN$ sent by $\calA$ during a secret key query and $P\subset\calN$ sent by $\calA$ during a ciphertext query.         
    \end{itemize}
\end{itemize}

    \end{mdframed}
    \vspace{-1mm}
    \caption{Vanilla Encryption Block Security Game}
    \vspace{-1mm}
    \label{fig:vanillasecurity}
\end{figure}

\paragraph{Decisional Diffie-Hellman.} Given a cyclic group $G$ and generator $g\in G$, the \emph{decisional Diffie-Hellman} (DDH) problem is the following.  Given $(a,b,c)\in G^3$ where $(a,b)=(g^x,g^y)$ for random exponents $x,y\sim\$$, and either:
\begin{itemize}
    \item[$\cdot$] $c=g^{xy}$; or
    \item[$\cdot$] $c=g^z$ for another random exponent $z\sim\$$;
\end{itemize}
determine which is the case.  We say that an algorithm \emph{breaks DDH} if it can solve this problem with probability at least $1/2+\epsilon$ for non-negligible $\epsilon>0$.

\paragraph{Rerandomizing DDH Triples.} We will need a common trick for using a DDH tuple $(a,b,c)$ to generate a ``fresh'' DDH tuple $(a,b',c')$ ``of the same type''.  Specifically, there is an efficient procedure which takes $(a=g^x,b=g^y,c)$ as input and outputs $(b'=g^{y'},c')$ such that $y'\sim\$$ is a uniformly random exponent and:
\begin{itemize}
    \item[$\cdot$] if $c=g^{xy}$ then $c'=g^{xy'}$;
    \item[$\cdot$] if $c\neq g^{xy}$ then $c'=g^{z'}$ for a uniformly random exponent $z'\sim\$$.
\end{itemize}
Such a procedure could work as follows: draw uniformly random exponents $r,s\sim\$$ and set \[b'=b^r\cdot g^{rs};\text{  }c'=c^r\cdot a^{rs}.\]

\paragraph{The ``Corrupt or Target'' DDH Game.} Our new security assumption is characterized by a security game between an adversary $\calA$ and challenger $\calC$.  Let $n\in\mathbb{N}$ be an integer parameter.
\begin{itemize}
    \item $\calC$ draws random exponents $x_1,\dots,x_n\sim\$$, sends $(g^{x_1},\dots,g^{x_n})$ to $\calA$, and chooses a random bit $b\sim\{0,1\}$.
    \item[$\bullet$ {\bf Corrupt or Target:}] The game now enters the ``corrupt or target'' phase and stays here until either $\calA$ chooses to proceed, or until every $i\in[n]$ has either been corrupted or targeted.  Specifically, a set $S$ is initialized to $[n]$, and then $\calA$ sends queries of the following type:
    \begin{itemize}
        \item[$\cdot$] $(i,\underline{\sf corrupt})$, for $i\in S$; in this case $\calC$ returns $x_i$ and removes $i$ from $S$, \emph{i.e.}, $S=S\setminus\{i\}$;
        \item[$\cdot$] $(i,\underline{\sf target})$, for $i\in S$; in this case $\calC$ chooses a random exponent $y_i\sim\$$ and sends $(g^{y_i},c_i)$ where $c_i=g^{x_iy_i}$ if $b=0$ and $c_i=g^{z_i}$ for another random exponent $z_i\sim\$$ if $b=1$; again $i$ is removed from $S$, \emph{i.e.}, $S=S\setminus\{i\}$.
    \end{itemize}
    \item Finally, $\calA$ sends a bit $b'\in\{0,1\}$ and wins if $b'=b$.
\end{itemize}

\begin{clm}
\label{clm:vanilla_block}
Suppose an efficient adversary $\calA$ breaks the security of the vanilla encryption block scheme above (\emph{i.e.}, $\calA$ wins the security game in Fig.~\ref{fig:vanillasecurity} with non-negligible advantage).  Then there exists another efficient adversary who can win the corrupt or target DDH game with non-negligible advantage.
\end{clm}

\begin{proof}
The adversary $\calB$ plays the corrupt or target DDH game against a challenger $\calC$.  Additionally, $\calB$ plays the vanilla encryption block security game as the challenger against the adversary $\calA$.  $\calB$ works as follows.
\begin{itemize}  
    \item On receiving $\{g^{\alpha_p}\}_{p\in\calN}$ from $\calC$, $\calB$ forwards $\{g^{\alpha_p}\}_{p\in\calN}$ to $\calA$.

    \item Now both games enter their respective query phases.  $\calB$ initializes the set $S=\calN$.
    \begin{itemize}
        \item[$\cdot$] Whenever $\calA$ sends the secret key query $p\in\calN$ to $\calB$, $\calB$ issues the query $(p,\underline{\sf corrupt})$ to $\calC$ and receives $\alpha_p$ from $\calC$ which it forwards to $\calA$; $\calB$ updates $S=S\setminus\{p\}$.
        \item[$\cdot$] Whenever $\calA$ sends the ciphertext query $(P,\msg)$ to $\calB$, $\calB$ does the following for each $p\in P$:
        \begin{itemize}
            \item[$-$] if $p\in S$, $\calB$ issues the query $(p,\underline{\sf target})$ to $\calC$ and receives the DDH challenge pair $(b_p,c_p)$;
            \item[$-$] if $p\notin S$ because $\calA$ has already issued a secret key query for $p$, $\calB$ immediately halts and ouputs a random bit;
            \item[$-$] if $p\notin S$ because $\calA$ has already issued a ciphertext query involving $p$, then $\calB$ has already obtained a challenge DDH pair $(b_p,c_p)$ for $p$;
            \item[$-$] at this point (assuming $\calB$ has not already halted), $\calB$ has a DDH challenge tuple $(g^{\alpha_p},b_p,c_p)$; $\calB$ next rerandomizes this tuple twice obtaining $(B_p,C_p)$ and $(B'_p,C'_p)$.
        \end{itemize}
        Finally, $\calB$ sends \[\bbB=\Big\{\bigl(B_p,C_p,B'_p,C'_p\cdot\msg\bigr)\Big\}_{p\in P}\] to $\calA$.
    \end{itemize}
    \item When $\calA$ indicates that the query phase is over by sending the bit $b'\in\{0,1\}$ to $\calB$, $\calB$ forwards this bit to $\calC$ and halts.
\end{itemize}
It can be checked that $\calB$ and $\calA$ are playing the many-way semantic security game with active corruptions using $\calC$'s secret bit $b$.  Thus, $\calB$ wins the corrupt or target DDH game whenever $\calA$ wins.  The claim follows.
\end{proof}

\subsection{Security of ABE with Payload Anonymity}
\label{appendix:FABEO}
\subsubsection{Quick Overview of FABEO}
In~\cite{FABEO}, a key-policy ABE scheme called FABEO is given.  We describe a very simplified version of the encryption/decryption mechanism of FABEO here for readers who are unfamiliar.

\vspace{.1mm}
\paragraph{The Encryption$/$Decryption Mechanism} We describe a stripped-down version of the scheme from FABEO~\cite{FABEO} in order to illustrate the encryption and decryption mechanics used in our payload anonymity scheme. FABEO builds a fraction of pairings to help decrypt the message and achieve efficiency. Let $\alpha\in\mathbb{Z}_p$ be the secret key and $e(g_1,g_2)^\alpha$ the public key.  Let $H:\calU\rightarrow\G_1$ be a hash function mapping attributes into $\G_1$.  Consider the policy which simply checks whether a single attribute $u^*\in\calU$ is satisfied.  The secret key for this policy will be $\bigl(g_1^\alpha\cdot H(u^*)^r,g_2^r\bigr)$; the ciphertext is $\bigl(H(u^*)^s,g_2^s,e(g_1,g_2)^{\alpha s}\cdot\msg\bigr)$.  Decryption works by computing \[\frac{e\bigl(g_1^\alpha\cdot H(u^*)^r,g_2^s\bigr)}{e\bigl(H(u^*)^s,g_2^r\bigr)}=e(g_1,g_2)^{\alpha s},\] and using this to recover $\msg$.

\subsubsection{Correctness of KP-ABE with Split KA and PA}
\label{appendix:correctness}
The correctness of KP-ABE with Split KA and PA shows as follows.  Consider each factor in the quantity ${\sf val}$ computed during decryption.  The numerator is \[e\bigl((\sk_{1,i}^{(j)})^{\gamma_i},\ct_2^{(j)}\bigr)=e\bigl(g_1^{\gamma_i\bM_i\cdot{\bf v}_j},g_2^{s_j}\bigr)\cdot e\bigl(H(\pi(i)),g_2\bigr)^{\gamma_ir_js_j},\] 
and the denominator is $e\bigl(H(\pi(i)),g_2\bigr)^{\gamma_ir_js_j}$.  Thus, we can cancel and we see that
\begin{align*}
    {\sf val}
    &=
    \prod_{i\in I_S}\prod_{j=1,2} e\bigl(g_1^{\gamma_i\bM_i\cdot{\bf v}_j},g_2^{s_j}\bigr)=\prod_{j=1,2} e\Bigl(g_1^{(\sum_i\gamma_i\bM_i)\cdot{\bf v}_j},g_2^{s_j}\Bigr)\\
    &=
    \prod_{j=1,2}e\bigl(g_1^\alpha,g_2^{s_j}\bigr)=e(g_1,g_2)^{\alpha s},
\end{align*}
and so decryption outputs $\msg$.

\subsubsection{Security Game}

We define the security game as Fig.~\ref{fig:abepasecurity}.
\begin{figure}[!h]
    \centering
    \begin{mdframed}[font=\small]
\begin{itemize}
    \item[{\bf 1. Setup:}] $\calC$ draws $(\mpk,\msk^{(1)},\msk^{(2)})\sim\setup(1^\lambda)$, and also a secret bit $b\sim\{0,1\}$, and sends $\mpk$ to $\calA$.
    
    \item[{\bf 2. Queries:}] The game now enters the query phase and remains here until $\calA$ chooses to proceed.  During each round, $\calA$ will issue one of the following types of queries to $\calC$ and $\calC$ will respond.
    \begin{itemize}
        \item[(i) \underline{Key Authority Query} $-$] $\calA$ sends $j\in\{1,2\}$ to $\calC$, who returns $\msk^{(j)}$.
        
        \item[(ii) \underline{Secret Key Queries} $-$] $\calA$ sends $\mathbb{A}$ to $\calC$, who returns $(\sk_{\mathbb{A}}^{(1)},\sk_{\mathbb{A}}^{(2)})$, where $\sk_{\mathbb{A}}^{(j)}\sim\keygen^{(j)}(\msk^{(j)},\mathbb{A})$, for $j=1,2$.
        
        \item[(iii) \underline{Ciphertext Queries} $-$] $\calA$ sends $(S_0,\msg_0),(S_1,\msg_1)$ to $\calC$ with $|S_0|=|S_1|$; $\calC$ draws $\ct\sim\enc(\mpk,S_b,\msg_b)$, where $\ct=(S_b,\ct_{\sf PL})$.  $\calC$ returns $\ct_{\sf PL}$ to $\calA$.
    \end{itemize}

    \item[{\bf 3. Guess:}] $\calA$ sends $b'\in\{0,1\}$ to $\calC$, and the game ends.
    
    \item[{\bf Outcome:}] $\calA$ wins if $b'=b$ and additionally if:
    \begin{itemize}
        \item[$-$] $\mathbb{A}(S)=0$ for all $\mathbb{A}$ and $S$ sent by $\calA$ during the secret key and ciphertext queries; and
        \item[$-$] $\calA$ asked at most one key authority query (\emph{i.e.}, $\calA$ did not receive both $\msk^{(1)}$ and $\msk^{(2)}$).  Moreover, if $\calA$ asked for one of $\msk^{(1)},\msk^{(2)}$, then $S_0=S_1$ for all ciphertext queries.        
    \end{itemize}
\end{itemize}

    \end{mdframed}
    \vspace{-2mm}
    \caption{Security Game for KP-ABE with Split KA and PA}
    \label{fig:abepasecurity}
\end{figure}

\subsubsection{Security Proof in the GGM}
\label{GGM_security}
In this section we prove that the KP-ABE scheme with split key-authority and payload anonymity is secure in the generic group model.  To get started, we simplify the security game.  Recall the game begins with $\calC$ sending $\calA$ $\mpk$ and then $\calA$ enters the query phase and can ask a key-authority query, secret key queries and ciphertext queries.  Since $\calA$ can only ask one key-authority query, we can selectively prepare and simplify the game a bit.  Specifically, our starting point is the following game:
\begin{itemize}
    \item $\calC$ chooses $\alpha,\beta,\delta\sim\$$ and sends $[\beta]_2$, $[\delta]_2$ and $[\alpha]_T$ to $\calA$.  Additionally, $\calC$ chooses a random bit $b\sim\{0,1\}$.
    \item Now $\calA$ enters the query phase and asks two types of queries:
    \begin{itemize}
        \item[$-$] $\calA$ can ask for a secret key query by sending an access structure $(\bM,\pi)$, and $\calC$ chooses a random ${\bf v}_1,{\bf v}_2\sim\mathbb{Z}_p^{n_2-1}$, $r_1,r_2\sim\mathbb{Z}_p$, and returns \[\Big\{\bigl[\bM_i(\alpha,{\bf v}_1)/\beta+ h_{\pi(i)}r_1/\beta\bigr]_1,\bigl[\bM_i(\alpha,{\bf v}_2)+h_{\pi(i)}r_2\bigr]_1,\]\[[h_{\pi(i)}r_1/\delta]_1, [h_{\pi(i)}r_2]_1\Big\}_{i\in[n_1]},\], and $[r_1]_2,[r_2]_2$.

        \item[$-$] $\calA$ can ask for a ciphertext query by sending $S\subset\calU$ and $\calC$ chooses $s_1,s_2\sim\mathbb{Z}_p$ and returns \[\Bigl(\big\{[h_us]_1\big\}_{u\in S},[s_1\beta]_2,[s_2]_2,[s_2\delta]_2,[s_1]_2\Bigr),\] and $[\alpha s]_T$ where $s=s_1+s_2$ if $b=0$ and $[z]_T$ for a random $z\sim\mathbb{Z}_p$ if $b=1$. 
    \end{itemize}
    \item Finally $\calA$ outputs $b'$ and wins if $b'=b$.
\end{itemize}

It is easy to show that if $\calA$ can win the payload anonymity game for the ABE scheme (see Fig.~\ref{fig:abepasecurity}) then it can also win this game.

     Now, we prove that no adversary can win this game in the GGM using the symbolic security framework of prior work~\cite{FABEO,AC17,ABGW17}.  For this purpose, we interpret $\calC$'s responses to $\calA$'s queries as polynomial evaluations.  GGM security then follows by establishing that the span of the query responses is disjoint from the span of the challenges.  This is because the GGM restricts how $\calA$ is allowed to manipulate group elements, so the only way $\calA$ could try to use the query responses to decide whether $\calC$ is sending $[\alpha s]_T$ or random is by linearly converting the query responses to something in the span of $\alpha s$ and then testing for equality; however this is impossible if the spans are disjoint. We will prove our scheme satisfies strong symbolic security from~\cite{FABEO} against an adversary who asks many times of ciphertext query and secret key query in order to derive GGM security. Now, let us zoom into the values that such an adversary receives from its queries.

First consider the level one encodings that $\calA$ gets for $k \in Q_{ct}, j \in Q_{sk}$:
\begin{itemize}
    \item[$\cdot$] $\Big\{\bigl[{\bf M}_i^{(j)}(\alpha,{\bf v}_1^{(j)})/\beta+h_{\pi(i)}^{(j)}r_1^{(j)}/\beta\bigr]_1\Big\}_{i\in[n_1]}$;
    \item[$\cdot$] $\Big\{\bigl[{\bf M}_i^{(j)}(\alpha,{\bf v}_2^{(j)})+h_{\pi(i)}^{(j)}r_2^{(j)}\bigr]_1\Big\}_{i\in[n_1]}$;
    \item[$\cdot$] $\Big\{\bigl[h_{\pi(i)}^{(j)}r_1^{(j)}/\delta\bigr]_1\Big\}_{i\in[n_1]}$;
    \item[$\cdot$] $\Big\{\bigl[h_{\pi(i)}^{(j)}r_2^{(j)}\bigr]_1\Big\}_{i\in[n_1]}$;
    \item[$\cdot$] $\Big\{\bigl[h_u^{(k)}s_1^{(k)}\bigr]_1\Big\}_{u\in S}$.
\end{itemize}
Next, consider all of the second level encodings $\calA$ gets for $k \in Q_{ct}, j \in Q_{sk}$: \[[\beta]_2,[\delta]_2,[r_1^{(j)}]_2,[r_2^{(j)}]_2,[s_1^{(k)}\beta]_2,[s_2^{(k)}]_2,[s_2^{(k)}\delta]_2,[s_1^{(k)}]_2.\]
We can then show
\begin{clm}
    For all $Q_{ct}, Q_{sk} \in \mathbb{N}$, such that $k \in Q_{ct}, j \in Q_{sk}$, we have \[\sfspan\big(\alpha\big(s_1^{(k)}+s_2^{(k)}\big)\big) \cap \sfspan\big(\Big\{\alpha, \Big\{1, \sigma\Big\}\otimes\Big\{1, \tau\Big\}\Big\}\big)=\{0\},\]
    where $\sigma$ is an encoding from the level one and $\tau$ is an encoding from the level two. 
\end{clm}
\begin{proof}
    The proof proceeds by contradiction. Suppose there exist some $z$ such that \[z=\sfspan\big(\alpha\big(s_1^{(k)}+s_2^{(k)}\big)\big) \cap \sfspan\big(\Big\{\alpha, \Big\{1, \sigma\Big\}\otimes\Big\{1, \tau\Big\}\Big\}\big),\]
    We can have 
    \begin{equation}\label{eq:z_left}
        z=\sum_{k}c_k\alpha\big(s_1^{(k)}+s_2^{(k)}\big)
    \end{equation}
    and 
    \begin{equation}\label{eq:z_right}
        z=d\alpha + \sum_{X\in\Big\{1, \sigma\Big\}, Y\in\Big\{1, \tau\Big\}}d_{XY}XY
    \end{equation}
    for some $c_k, d, d_{X,Y}$.

    Since for all ciphertext queries and secret key queries, we have $\mathbb{A}(S)=0$. That means there exists some ${\bf w}$ such that $\bM_i{\bf w}=0$ for all $i$ such that $\pi(i)\in S$, and moreover the first coordinate of ${\bf w}$ is non-zero.  So consider evaluating $(\alpha,v_1)$ and $(\alpha,v_2)$ to the multiple of ${\bf w}$ such that the first coordinate is $\alpha$. Simplifying in this way gives us $\sigma$ equals to:
    \begin{itemize}
    \item[$\cdot$] $\Big\{\bigl[h_u^{(j)}r_1^{(j)}/\beta\bigr]_1,\bigl[h_u^{(j)}r_2^{(j)}\bigr]_1,\bigl[h_u^{(j)}r_1^{(j)}/\delta\bigr]_1,[h_u^{(k)}s_1^{(k)}]_1\Big\}_{u\in S}$;
    \item[$\cdot$] $\Big\{\bigl[\bM_i^{(j)}(\alpha,{\bf v}_1^{(j)})/\beta+h_{\pi(i)}^{(j)}r_1^{(j)}/\beta\bigr]_1\Big\}_{i:\pi(i)\notin S}$;
    \item[$\cdot$] $\Big\{\bigl[\bM_i^{(j)}(\alpha,{\bf v}_2^{(j)})+h_{\pi(i)}^{(j)}r_2^{(j)}\bigr]_1\Big\}_{i:\pi(i)\notin S}$.
    \end{itemize}
    Thus we can consider $i:\pi(i)\in S$ and $i:\pi(i)\notin S$ separately.

    \paragraph{\underline{\sf Step 1:}}
    We first consider $i:\pi(i)\in S$, thus we have \[\sigma\in\Big\{\bigl[h_u^{(j)}r_1^{(j)}/\beta\bigr]_1,\bigl[h_u^{(j)}r_2^{(j)}\bigr]_1,\bigl[h_u^{(j)}r_1^{(j)}/\delta\bigr]_1[h_u^{(k)}s_1^{(k)}]_1\Big\}_{u\in S}\]

    Then we can evaluate $X\in\{1,\sigma\}$ and get Equation~\ref{eq:z_right} equals to:
    \begin{multline*}
        d\alpha + \sum_{j_1,j_2,k}e_{j_1,j_2,k}h_u^{(j_1)}r_1^{(j_1)}\beta^{-1}\cdot Y^{(j_2,k)} \\
        +  \sum_{j_1,j_2,k}f_{j_1,j_2,k}h_u^{(j_1)}r_2^{(j_1)}\cdot Y^{(j_2,k)} \\
        +      \sum_{j_1,j_2,k}g_{j_1,j_2,k}h_u^{(j_1)}r_1^{(j_1)}\delta^{-1}\cdot Y^{(j-2,k)} 
        +   \\\sum_{j,k_1,k_2}O_{j,k_1,k_2}h_u^{(k_1)}s_1^{(k_1)}\cdot Y^{(j,k_2)} 
        + \sum_{j,k}p_{j,k} Y^{(j,k)}
    \end{multline*}
    
    Equation~\ref{eq:z_left} only contains the product of $\alpha s_1^{(k)}$ and $\alpha  s_2^{(k)}$. However, there is no $\alpha$ contained in $Y$. Althrough $Y$ contains $s_1$ and $s_2$, no product of $Y$ contains $\alpha$ in our reduction of Equation~\ref{eq:z_right}. We can know Equation~\ref{eq:z_right} does not have the product of $\alpha s_1$ and $\alpha  s_2$. Thus, Equation~\ref{eq:z_left} and Equation~\ref{eq:z_right} can not be equal when $X\in\{1,\sigma\}$.

    \paragraph{\underline{\sf Step 2:}} We then consider $i:\pi(i)\notin S$, thus we have \[\sigma\in\Big\{U^{(j)},V^{(j)}\Big\}_j\] where 
    \begin{itemize}
    \item[$\cdot$] $U^{(j)}\in\Big\{\bigl[\bM_i^{(j)}(\alpha,{\bf v}_1^{(j)})/\beta+h_{\pi(i)}^{(j)}r_1^{(j)}/\beta\bigr]_1\Big\}_{i:\pi(i)\notin S}$; \\
    \item[$\cdot$] $V^{(j)}\in\Big\{\bigl[\bM_i^{(j)}(\alpha,{\bf v}_2^{(j)})+h_{\pi(i)}^{(j)}r_2^{(j)}\bigr]_1\Big\}_{i:\pi(i)\notin S}$.
    \end{itemize}
    Then we can evaluate $X\in\{1,\sigma\}$ and get Equation~\ref{eq:z_right} equals to
    \begin{multline*}
        d\alpha + \sum_{j_1,j_2,k}e_{j_1,j_2,k}U^{(j_1)}\cdot Y^{(j_2,k)} 
        + \sum_{j_1,j_2,k}f_{j_1,j_2,k}V^{(j_1)}\cdot Y^{(j_2,k)} \\
        + \sum_{j}g_{j}U^{(j)} + \sum_{j}l_{j}V^{(j)}
    \end{multline*}
    Since  Equation~\ref{eq:z_left} only contains the product of $\alpha s_1^{(k)}$ and $\alpha  s_2^{(k)}$, we can cancel the term that does not contains $\alpha s_1^{(k)}$ or $\alpha  s_2^{(k)}$. Then we get
    \begin{multline*}
        \sum_{j,k}e_{j,k}\bM_i^{(j)}(\alpha,{\bf v}_1^{(j)})\cdot s_1^{(k)} + \sum_{j,k}e_{j,k}(h_{\pi(i)}^{(j)}r_1^{(j)})\cdot s_1^{(k)} \\ 
        + \sum_{j,k}e_{j,k}\bM_i^{(j)}(\alpha,{\bf v}_1^{(j)})\cdot s_2^{(k)} / \beta + \sum_{j,k}e_{j,k}(h_{\pi(i)}^{(j)}r_1^{(j)})\cdot s_2^{(k)} / \beta \\
        + \sum_{j,k}e_{j,k}\bM_i^{(j)}(\alpha,{\bf v}_1^{(j)})\cdot s_2^{(k)} \delta/ \beta + \sum_{j,k}e_{j,k}(h_{\pi(i)}^{(j)}r_1^{(j)})\cdot s_2^{(k)}\delta / \beta \\
        + \sum_{j,k}e_{j,k}\bM_i^{(j)}(\alpha,{\bf v}_1^{(j)})\cdot s_1^{(k)} / \beta + \sum_{j,k}e_{j,k}(h_{\pi(i)}^{(j)}r_1^{(j)})\cdot s_1^{(k)} / \beta \\
        + \sum_{j,k}f_{j,k}\bM_i^{(j)}(\alpha,{\bf v}_2^{(j)})\cdot s_1^{(k)}\beta + \sum_{j,k}f_{j,k}(h_{\pi(i)}^{(j)}r_2^{(j)})\cdot s_1^{(k)}\beta\\
        + \sum_{j,k}f_{j,k}\bM_i^{(j)}(\alpha,{\bf v}_2^{(j)})\cdot s_2^{(k)} + \sum_{j,k}f_{j,k}(h_{\pi(i)}^{(j)}r_2^{(j)})\cdot s_2^{(k)}\\
        + \sum_{j,k}f_{j,k}\bM_i^{(j)}(\alpha,{\bf v}_2^{(j)})\cdot s_2^{(k)}\delta + \sum_{j,k}f_{j,k}(h_{\pi(i)}^{(j)}r_2^{(j)})\cdot s_2^{(k)}\delta\\
        + \sum_{j,k}f_{j,k}\bM_i^{(j)}(\alpha,{\bf v}_2^{(j)})\cdot s_1^{(k)} + \sum_{j,k}f_{j,k}(h_{\pi(i)}^{(j)}r_2^{(j)})\cdot s_1^{(k)}\\
    \end{multline*}
    ${\bf v}_1^{(j)}, {\bf v}_2^{(j)}$, $\beta$, $\beta^{-1}$,$r_1$.$r_2$ and $\delta
    $ cannot be canceled, thus Equation~\ref{eq:z_left} does not contains the product of $\alpha s_1^{(k)}$ and $\alpha  s_2^{(k)}$. Equation~\ref{eq:z_left} and Equation~\ref{eq:z_right} cannot be equal when $i:\pi(i)\notin S$.

    For both $i:\pi(i)\in S$ and $i:\pi(i)\notin S$, Equation~\ref{eq:z_left} and Equation~\ref{eq:z_right} cannot be equal. Thus, we show the contradiction. 
\end{proof}

This shows that our scheme has strong symbolic security, which by the arguments of~\cite{FABEO}, implies that it is secure in the GGM.

\subsection{Anonymity Proof in the GGM}
Similar as the security proof, we prove the anonymity of our scheme also in GGM. Consider the same game that used for security proof. Note that our anonymity adversary model intend to prevent attackers from outside the system, thus we only allow $\calA$ do ciphertext queries. The game them can be simplify as follows:
\begin{itemize}
    \item $\calC$ chooses $\alpha,\beta_1,\beta_2,\delta_1,\delta_2\sim\$$ and sends $[\beta_1]_2$, $[\beta_2]_2$, $[\delta_1]_2$,$[\delta_2]_2$ and $[\alpha]_T$ to $\calA$.  Additionally, $\calC$ chooses a random bit $b\sim\{0,1\}$.
    \item Now $\calA$ enters the query phase and can ask for a ciphertext query by sending $S\subset\calU$ and $\calC$ chooses $s_1,s_2\sim\mathbb{Z}_p$ and returns \[\Bigl(\big\{[h_us]_1\big\}_{u\in S},[s_1\beta_1]_2,[s_2\beta_2]_2,[s_2\delta_1]_2,[s_1\delta_2]_2\Bigr),\] and $[\alpha s]_T$ where $s=s_1+s_2$ if $b=0$ and $[z]_T$ for a random $z\sim\mathbb{Z}_p$ if $b=1$. 
    \item Finally $\calA$ outputs $b'$ and wins if $b'=b$.
\end{itemize}

Now, we prove that no adversary can win this game in GGM as security proof in \ref{GGM_security}. Let's consider the level one encodings that $\calA$ gets for $k\in Q_{ct}$:
\begin{itemize}
    \item[$\cdot$] $\Big\{\bigl[h_u^{(k)}s^{(k)}\bigr]_1\Big\}_{u\in S}$;
    \item[$\cdot$] $[h_u^{(k)}]_1$
\end{itemize}
Note that $\calA$ would also calculate $[h_u^{(k)}]_1$ by itself.

Next, consider all of the second level encodings $\calA$ gets for $k \in Q_{ct}, j \in Q_{sk}$: \[[\beta_1]_2,[\beta_2]_2,[\delta_1]_2,[\delta_2]_2,[s_1^{(k)}\beta_1]_2,[s_2^{(k)}\beta_2]_2,[s_2^{(k)}\delta_1]_2,[s_1^{(k)}\delta_2]_2.\]
We can then show
\begin{clm}
    For all $Q_{ct} \in \mathbb{N}$, such that $k \in Q_{ct}$, we have \[\sfspan\big(\alpha\big(s_1^{(k)}+s_2^{(k)}\big)\big) \cap \sfspan\big(\Big\{\alpha, \Big\{1, \sigma\Big\}\otimes\Big\{1, \tau\Big\}\Big\}\big)=\{0\},\]
    where $\sigma$ is an encoding from the level one and $\tau$ is an encoding from the level two. 
\end{clm}
\begin{proof}
    The proof proceeds by contradiction. Suppose there exist some $z$ such that \[z=\sfspan\big(\alpha\big(s_1^{(k)}+s_2^{(k)}\big)\big) \cap \sfspan\big(\Big\{\alpha, \Big\{1, \sigma\Big\}\otimes\Big\{1, \tau\Big\}\Big\}\big),\]
    We can have 
    \begin{equation}\label{eq:z2_left}
        z=\sum_{k}c_k\alpha\big(s_1^{(k)}+s_2^{(k)}\big)
    \end{equation}
    and 
    \begin{equation}\label{eq:z2_right}
        z=d\alpha + \sum_{X\in\Big\{1, \sigma\Big\}, Y\in\Big\{1, \tau\Big\}}d_{XY}XY
    \end{equation}
    for some $c_k, d, d_{X,Y}$.

    Then we can evaluate $X\in\{1,\sigma\}$ and get Equation~\ref{eq:z2_right} equals to:    
    \begin{multline*}
        d\alpha + \sum_{k}\beta_1(e_{k,1}h_u^{(k)}+e_{k,2}h_u^{(k)}s^{(k)}) + \sum_{k}\beta_2(f_{k,1}h_u^{(k)}+f_{k,2}h_u^{(k)}s^{(k)}) \\ + \sum_{k_1k_2}\beta_1s_1^{(k_2)}(g_{k_1,k_2,1}h_u^{(k)}+g_{k_1,k_2,2}h_u^{(k)}s^{(k)}) + \\\sum_{k_1k_2}\beta_2s_2^{(k_2)}(l_{k,1}h_u^{(k)}+l_{k,2}h_u^{(k)}s_1^{(k)}) + \\\sum_{k}\delta_1(m_{k,1}h_u^{(k)}+m_{k,2}h_u^{(k)}s_1^{(k)}) + \sum_{k}\delta_2(n_{k,1}h_u^{(k)}+n_{k,2}h_u^{(k)}s_1^{(k)}) + \\\sum_{k_1k_2}\delta_1s_2^{(k_2)}(o_{k_1,k_2,1}h_u^{(k_1)}+o_{k_1,k_2,2}h_u^{(k_1)}s_1^{(k_2)}) + \\\sum_{k_1k_2}\delta_2s_1^{(k_2)}(p_{k_1,k_2,1}h_u^{(k_1)}+p_{k_1,k_2,2}h_u^{(k_1)}s_1^{(k_1)})  
    \end{multline*}
    We can group $\beta_1$ as an example, and $\beta_2$, $\delta_1$ and $\delta_2$ will have similar result. 
    \begin{multline*}           
        \sum_{k_1,k_2}\beta_1(e_{k_1,1}h_u^{(k_1)}+e_{k_1,2}h_u^{(k_1)}s^{(k_1)} + \\g_{k_1,k_2,1}h_u^{(k)}s_1^{(k_2)}+g_{k_1,k_2,2}h_u^{(k)}s^{(k)}s_1^{(k_2)} \\
    \end{multline*}
    Since the left side only contains $s_1+s_2$,  $e_{k_1,1}h_u^{(k_1)}$ and $g_{k_1,k_2,2}h_u^{(k)}s^{(k)}s_1^{(k_2)}$ has to be 0, thus we can cancle them. The formula can then be wirte as
    \begin{multline*}           
        \sum_{k_1,k_2}\beta_1(e_{k_1,2}h_u^{(k_1)}s^{(k_1)} + g_{k_1,k_2,1}h_u^{(k)}s_1^{(k_2)}) \\
        = \sum_{k_1,k_2}\beta_1((e_{k_1,2}h_u^{(k_1)} + g_{k_1,k_2,1}h_u^{(k)})s_1^{(k_2)} + e_{k_1,2}h_u^{(k_1)}s_2^{(k_2)}
    \end{multline*}
    It's easy to show this formula does not have $s_1+s_2$, and the same as the following parameters.

\end{proof}

\subsection{Security of \name}
\begin{figure*}[t]
\begin{minipage}[b]{.5\linewidth}
    \begin{minipage}[b]{.5\linewidth}
        \captionsetup{justification=centering}
        \centering
         \includegraphics[width=\linewidth]{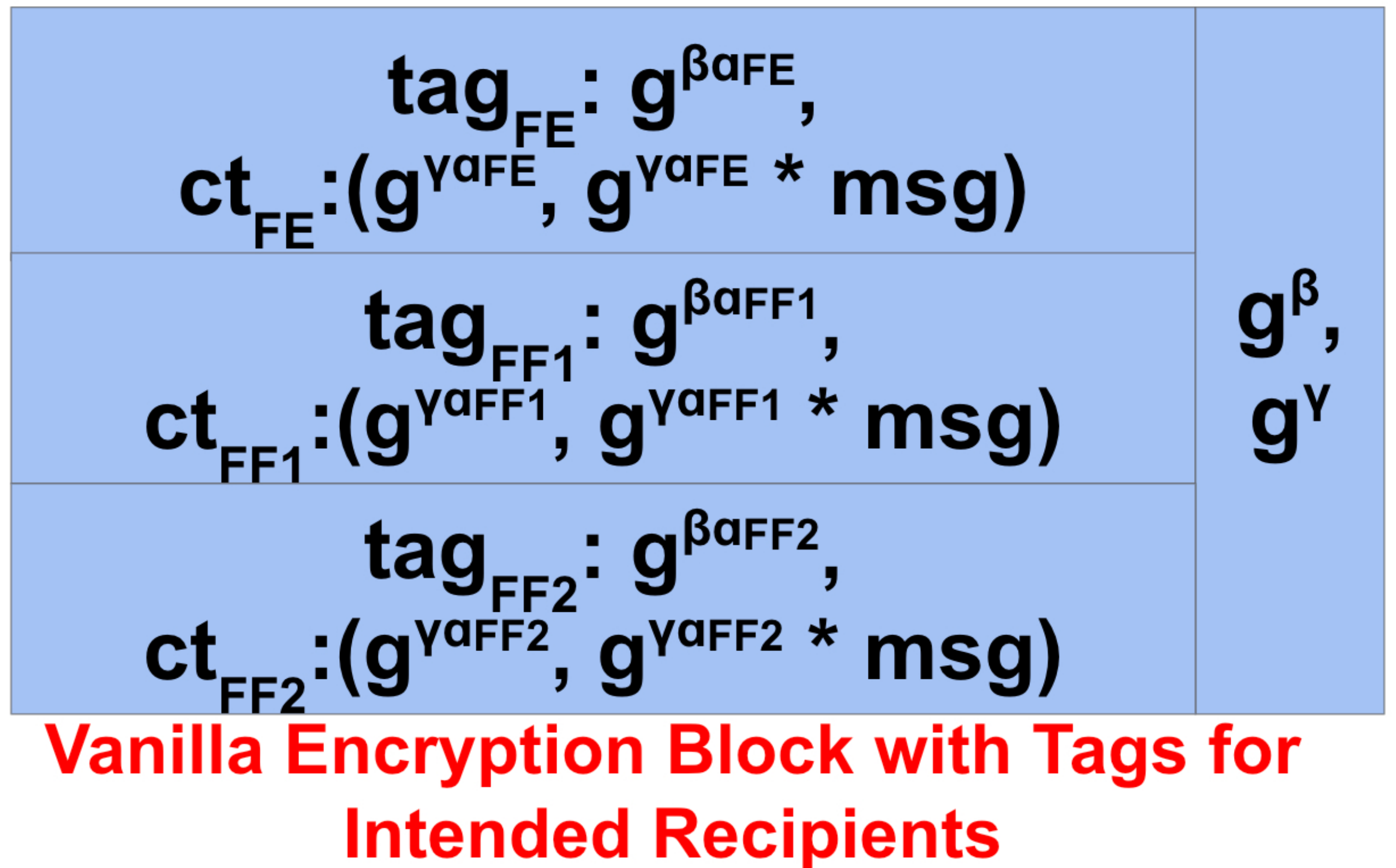}
        \subcaption{{\bf Diffie-Hellman based Construction }} 
        \label{fig:ABEAT_detail_VB}
    \end{minipage}\hfill
    \begin{minipage}[b]{.5\linewidth}
        \captionsetup{justification=centering}
        \centering
        \includegraphics[width=\linewidth]{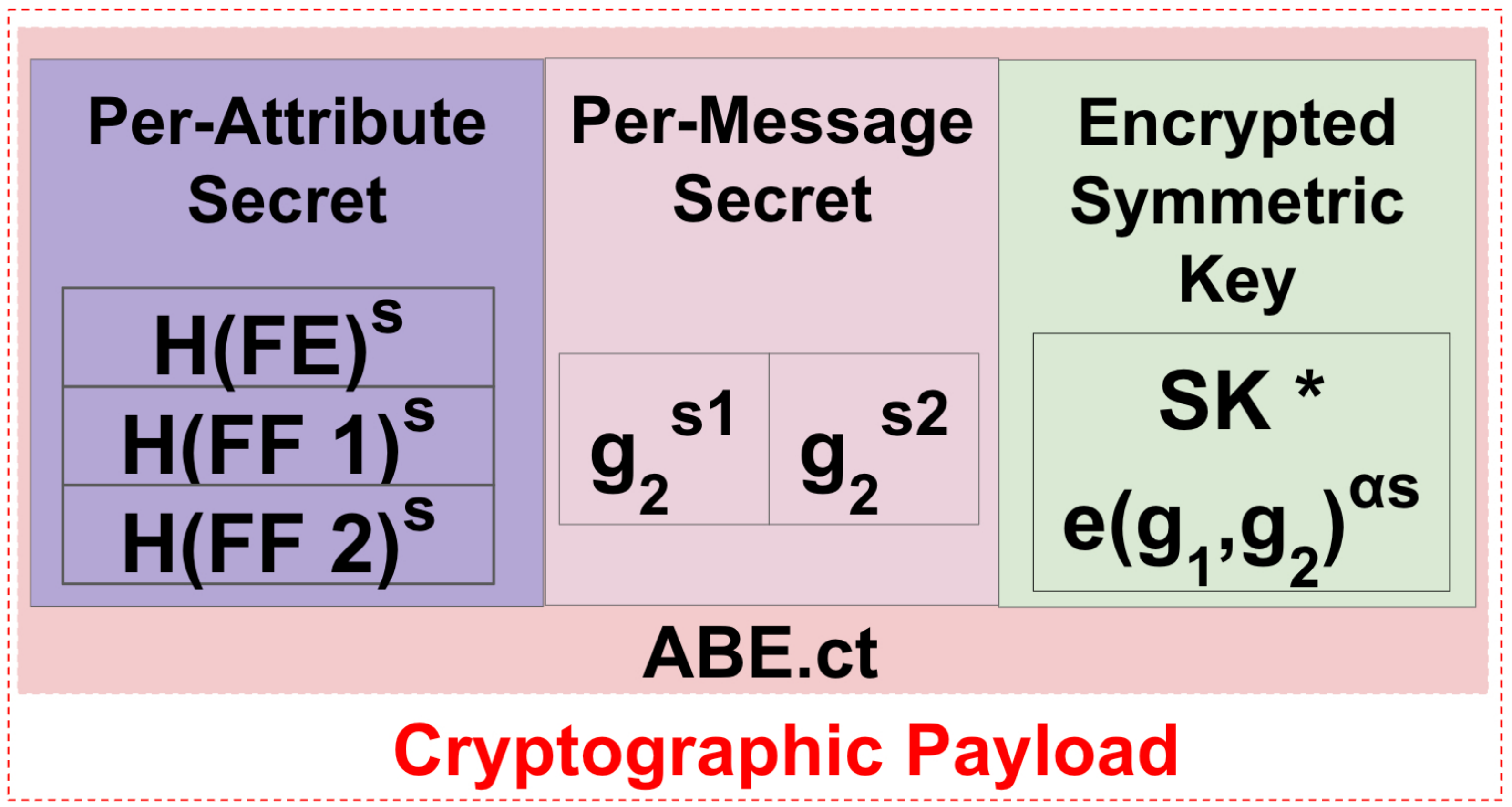}
        \vspace{.01mm}
        \subcaption{{\bf ABE based Construction}} 
        \label{fig:ABEAT_detail_CP}
    \end{minipage}
    \vspace{-5mm}
    \caption{{\bf Detailed Construction of \name}}
    \vspace{-2mm}
    \label{fig:ABEAT_detail}
\end{minipage}
\begin{minipage}[b]{.49\linewidth}
\centering
\includegraphics[width=\linewidth]{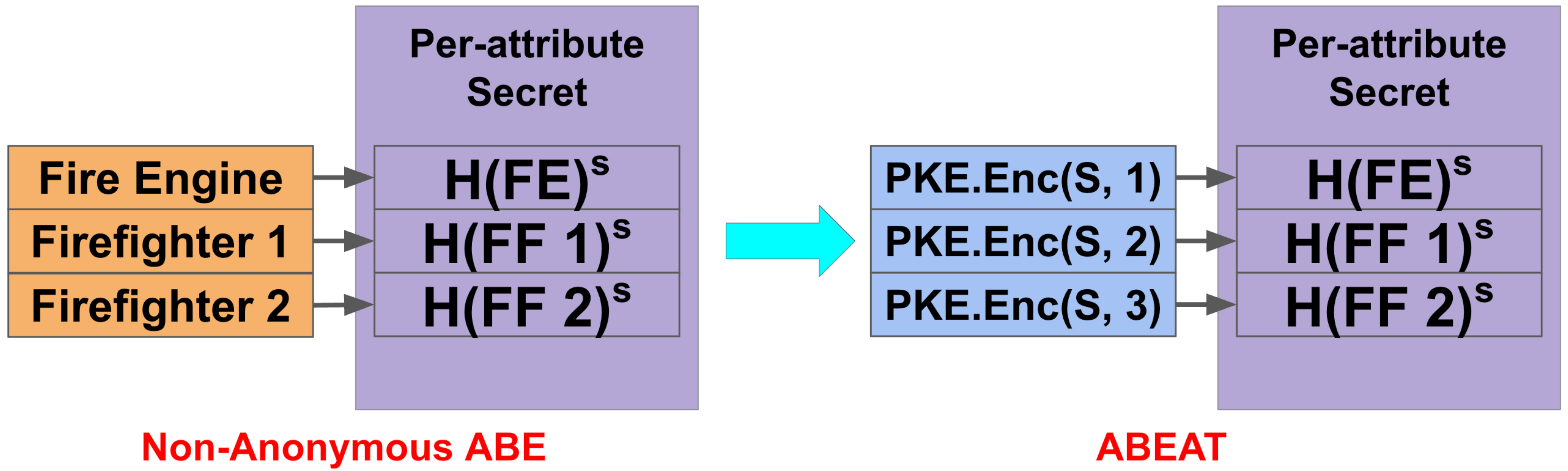}
\vspace{4mm}
\caption{{\bf Changing of Secret Structure}}
\vspace{-2mm}
\label{fig:index}
\end{minipage}
\end{figure*}
{\bf Security Game}
    We define the security game for \name as Fig.~\ref{fig:abegmsecurity}.
    \begin{figure}[!h]
        \centering
        \begin{mdframed}[font=\small]
    \begin{itemize}
        \item[{\bf 1. Setup:}] $\calC$ draws a secret bit $b\sim\{0,1\}$ and also \[\Bigl(\mpk,\msk^{(1)},\msk^{(2)},\big\{(\pk_p,\sk_p)\big\}_{p\in\calN}\Bigr)\sim\setup(\calN),\] and sends $\bigl(\mpk,\{\pk_p\}_{p\in\calN}\bigr)$ to $\calA$.
        
        \item[{\bf 2. Queries:}] The game now enters the query phase and remains here until $\calA$ chooses to proceed.  During each round, $\calA$ will issue one of the following types of queries to $\calC$ and $\calC$ will respond.
        \begin{itemize}
            \item[(i) \underline{Key Authority Query} $-$] $\calA$ sends $j\in\{1,2\}$ to $\calC$, who returns $\msk^{(j)}$.
            
            \item[(ii) \underline{Secret Key Queries} $-$] $\calA$ sends $\mathbb{A}$ to $\calC$, who returns $(\sk_{\mathbb{A}}^{(1)},\sk_{\mathbb{A}}^{(2)})$, where $\sk_{\mathbb{A}}^{(j)}\sim\keygen^{(j)}(\msk^{(j)},\mathbb{A})$, for $j=1,2$.
            
            \item[(iii) \underline{Ciphertext Queries} $-$] $\calA$ sends $(S_0,\msg_0),(S_1,\msg_1)$ to $\calC$ where $|S_0|=|S_1|$ and where $S_0$ and $S_1$ have the intended recipient sets $P_0,P_1\subset\calN$, respectively, where $|P_0|=|P_1|$.  $\calC$ draws \[\ct\sim\enc\Bigl(\mpk,S_b,\{\pk_p\}_{p\in P_b},\msg_b\Bigr)\] and sends $\ct$ to $\calA$.
    
            \item[(iv) \underline{Name Queries} $-$] $\calA$ sends $p\in\calN$ to $\calC$, who returns $\sk_p$.
        \end{itemize}
    
        \item[{\bf 3. Guess:}] $\calA$ sends $b'\in\{0,1\}$ to $\calC$, and the game ends.
        
        \item[{\bf Outcome:}] $\calA$ wins if b'=b and additionally if:
        \begin{itemize}
            \item[$-$] $\mathbb{A}(S)=0$ for all $\mathbb{A}$ and $S$ sent by $\calA$ during the secret key and ciphertext queries; and
            \item[$-$] $\calA$ asked at most one key authority query (\emph{i.e.}, $\calA$ did not receive both $\msk^{(1)}$ and $\msk^{(2)}$).  Moreover, if $\calA$ asked for one of $\msk^{(1)},\msk^{(2)}$, then $S_0=S_1$ for all ciphertext queries
            
            \item[$-$] for all $(S_0,\msg_0),(S_1,\msg_1)$ sent by $\calA$ as a ciphertext query and all $p\in\calN$ sent by $\calA$ as a name query, if $p\in P_0\cup P_1$ (where $P_i$ is the intended recipient set specified in $S_i$), then $S_0=S_1$.
        \end{itemize}
    \end{itemize}
    
        \end{mdframed}
        \vspace{-3mm}
        \caption{Security Game for KP-ABE with Split KA and A for GM}
        \vspace{-2mm}
        \label{fig:abegmsecurity}
    \end{figure}

\begin{thm}
\label{thm:abeat}
Assume $\veb$ is a scheme for generating a block of vanilla encryptions, and $\abe$ is a KP-ABE scheme with split key-authority and payload anonymity.  Then $\name$ is a KP-ABE scheme with split key-authority and anonymity for group communication.  Furthermore, the decryption costs for intended recipients and unintended recipients, respectively, are ${\sf T}_{\veb.\dec}+{\sf T}_{\abe.\dec}$ and ${\sf T}_{\veb.\dec}$, where ${\sf T}_{\veb.\dec}$ and ${\sf T}_{\abe.\dec}$ are decryption times of $\veb$ and $\abe$.
\end{thm}

\paragraph*{Remark.} Our implementations of the subroutines in the next section will obtain ${\sf T}_{\abe.\dec}\gg{\sf T}_{\veb.\dec}$, meaning that recipient verification is efficient compared to the decryption cost of intended recipients.

\begin{proof}
It is clear that $\name$ has the required syntax and correctness follows immediately from correctness of $\veb$ and $\abe$.  The claims about the runtime of decryption are clear by inspecting $\name$'s decryption procedure.  Therefore, it suffices to prove security.  Suppose an efficient adversary $\calA$ plays $\mathcal{G}$, the security game for KP-ABE with split key-authority and anonymity for group messaging (see Fig.~\ref{fig:abegmsecurity}).  Consider the hybrid game $\mathcal{H}$ which is the same as $\mathcal{G}$ except for how the challenger answers the ciphertext queries.  Upon receiving $(S_0,\msg_0),(S_1,\msg_1)$, $\calC$ checks whether $S_0=S_1$.  If equality holds, then $\calC$ responds to the query just as it would in $\mathcal{G}$.  If $S_0\neq S_1$, then $\calC$ responds with $\ct=(\bbB,\ct_{\sf PL})$ prepared as follows:
\begin{itemize}
    \item[$\cdot$] $(S_b,\ct_{\sf PL})\sim\abe.\enc(\mpk,S_b,\msg_b)$ is drawn;
    \item[$\cdot$] $\bbB\sim\$$ is drawn uniformly at random.
\end{itemize}

The theorem follows from Claims~\ref{clm:hybrid1} and~\ref{clm:hybrid2} below.
\end{proof}

\vspace{-4mm}
\begin{clm}
    \label{clm:hybrid1}
    For any PPT adversary $\calA$, \[\Big|{\rm Pr}\bigl[\calA\text{ wins }\mathcal{G}\bigr]-{\rm Pr}\bigl[\calA\text{ wins }\mathcal{H}\bigr]\Big|={\sf negl}(\lambda).\]
\end{clm}

\begin{proof}
    This is by reduction to the security of $\veb$.  Specifically, consider the following adversary $\calB$ who plays the security game of $\veb$ (see Fig.~\ref{fig:vanillasecurity}).  On receiving $\{\pk_p\}_{p\in\calN}$ from $\calC$, $\calB$ runs $\abe.\setup(1^\lambda)$ and invokes $\calA$, sending $\mpk$ and $\{\pk_p\}_{p\in\calN}$.  $\calB$ then chooses $b\sim\{0,1\}$ and enters the query phase with $\calA$, responding to $\calA$'s queries as follows:
    \begin{itemize}
        \item $\calB$ responds to key-authority and secret-key queries honestly using $(\msk^{(1)},\msk^{(2)})$.

        \item $\calB$ responds to the name query $p\in\calN$ by forwarding $p$ to $\calC$, receiving $\sk_p$ and forwarding $\sk_p$ back to $\calA$.

        \item On receiving the ciphertext query $(S_0,\msg_0),(S_1,\msg_1)$ from $\calA$, $\calB$ returns $(\bbB,\ct_{\sf PL})$, prepared as follows:
        \begin{itemize}
            \item[$-$] if $S_0=S_1$, then let $S=S_0=S_1$ and let $P\subset\calN$ be the intended recipient set inside $S$.  $\calB$ draws: 
            \begin{itemize}
                \item[$\cdot$] $\bbB\sim\veb.\enc\bigl(P,\{\pk_p\}_{p\in P},S\bigr)$; and
                \item[$\cdot$] $(S,\ct_{\sf PL})\sim\abe.\enc(\mpk,S,\msg_b)$.
            \end{itemize}

            \item[$-$] if $S_0\neq S_1$, $\calB$ draws $(S_b,\ct_{\sf PL})\sim\abe.\enc(\mpk,S_b,\msg_b)$ and sends the ciphertext query $(P_b,S_b)$ to $\calC$ where $P_b\subset\calN$ is the intended recipient set inside $S_b$, receiving $\bbB$.
        \end{itemize}
    \end{itemize}
    Finally, $\calB$ forwards $\calA$'s guess bit to $\calC$.  When $\calC$ answers $\calB$'s ciphertext queries honestly, $\calA$ is playing $\mathcal{G}$, while when $\calC$ is answering $\calB$'s ciphertext queries with random strings, $\calA$ is playing $\mathcal{H}$.  Thus, the claim follows from security of $\veb$. 
\end{proof}

\vspace{-3mm}
\begin{clm}
    \label{clm:hybrid2}
    \vspace{-1mm}
    For any PPT adversary $\calA$, \[{\rm Pr}\bigl[\calA\text{ wins }\mathcal{H}\bigr]=\frac{1}{2}+{\sf negl}(\lambda).\]
\end{clm}
\begin{proof}
    This is by reduction to the security of $\abe$.  Specifically, consider the following adversary $\calB$ who plays the security game for $\abe$ (see Fig.~\ref{fig:abegmsecurity}).  On receiving $\mpk$ from $\calC$, $\calB$ runs $\veb.\keygen(\calN)$ and invokes $\calA$, sending $\mpk$ and $\{\pk_p\}_{p\in\calN}$.  $\calB$ then enters the query phase with $\calA$, responding to $\calA$'s queries as follows:
    \begin{itemize}
        \item $\calB$ responds to the name query $p\in\calN$ by sending $\sk_p$ to $\calA$.
        
        \item $\calB$ handles key-authority and secret-key queries by forwarding the query to $\calC$ and forwarding $\calC$'s response back to $\calA$.

        \item $\calB$ responds to the ciphertext query $(S_0,\msg_0),(S_1,\msg_1)$ with $(\bbB,\ct_{\sf PL})$, prepared as follows:
        \begin{itemize}
            \item[$-$] $\calB$ forwards the ciphertext query to $\calC$ and receives $\ct_{\sf PL}$ back;
            \item[$-$] if $S_0\neq S_1$, then $\calB$ draws $\bbB\sim\$$ uniformly at random;
            \item[$-$] if $S_0=S_1$, then let $S=S_0=S_1$ and let $P\subset\calN$ be the intended recipient set inside $S$.  $\calB$ draws $\bbB\sim\veb.\enc\bigl(P,\{\pk_p\}_{p\in P},S\bigr)$.
        \end{itemize}
    \end{itemize}
Finally, when $\calA$ outputs its guess bit $b'$, $\calB'$ simply copies $b'$ to $\calC$.  If $\calA$ wins $\mathcal{H}$ with non-negligible advantage, then $\calB$ breaks the security of $\abe$.
\end{proof}

\section{Supplemental Material}
\subsection{How \name Meets the Requirements of Special Operations}
We summarize the key features of \name that make it suitable for secure communication among dynamically formed groups that is suitable for a number of situations allowing subsequent verification and auditing of the information exchanged. 
\paragraph*{Anonymization and fast recipient verification}
\name provides recipient anonymization and fast verification through the Vanilla Encryption Block.
This protects the identity of the recipients during
transmission. In addition, fast verification lets unintended recipients efficiently determine whether a message is for them without heavy computation. Together, these features
significantly reduce system overhead and allows \name the flexibility to use a variety of underlying
communication channels (\eg multicast or broadcast) without incurring high cost for each receiver.
\paragraph*{Group messaging on a per-message basis}
\name supports group communication on a per-message basis, eliminating the need for explicit and
costly group formation protocol. For example, with MLS~\cite{mls}, establishing a group typically requires multiple
rounds of interaction among members and relies on strict message ordering and consistency across
participants---a requirement that is difficult to satisfy in disruption-prone networks.
In contrast, \name defines the message group at the time of encryption by the sender,
allowing communication without having to maintain information of the static group. This makes \name well-suited for
highly dynamic environments where membership changes frequently and where the sequence of message arriving is not guaranteed.
\paragraph*{Split key authorities to mitigate risk of compromised keys}
\name employs a split-authority design to reduce the security risk associated with compromised keys.
No single authority holds enough information to decrypt ciphertexts, preventing unilateral abuse or
leakage of a single master secret key. However, when auditing is required, a higher-level authority may combine
the partial master secret keys held by each authority to decrypt and review messages as necessary. This
design balances security against accountability, providing operational security without sacrificing
auditability.

\begin{figure}[t]
\centering
\includegraphics[width=\linewidth]{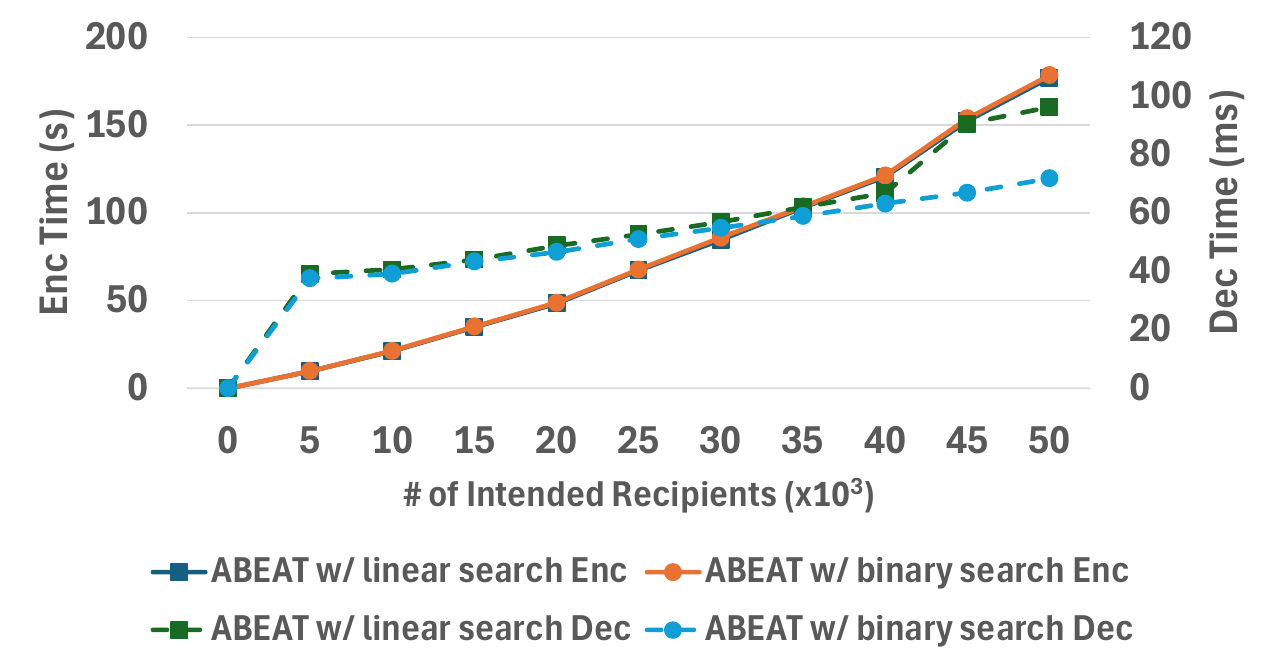}
\caption{{\bf \name w/ Binary Search}}
\label{fig:ABEAT_BS}
\end{figure}

\subsection{Implementation Details}
Fig.~\ref{fig:ABEAT_detail} shows an example construction of \name that has ``Fire Engine'',  ``Firefighter 1'', and ``Firefighter 2'' as intended recipients. An important part of $\abe.\ct$ is the set of per-attribute secrets, since recipients are required to locate the corresponding per-attribute secret to enable decryption.
As Fig.~\ref{fig:index} shows, in non-anonymous KP-ABE, $S$ will indicate the set of per-attribute secrets. However, in \name, $S$ is concealed, thus making it challenging for a recipient to identify the corresponding secret. To address this, while maintaining efficiency, we implement the vanilla encryption block by also encrypting an index that is used to map to the corresponding secret for each intended recipient.

\subsection{Performance of \name with Binary Search}
\label{sec:BS}
When the number of entries in the vanilla encryption block is large, a linear search for the tag could be inefficient. 
One solution is to use binary search instead of linear search when the number of entries in the vanilla encryption block is large. We use the x-coordinate as the key for binary search.

\paragraph*{Encryption} 
We insert entries into the vanilla encryption block in sorted order, based on the tag. The solid line using the left y-axis in Fig.~\ref{fig:ABEAT_BS} shows the encryption time of \name for both linear and binary search. The additional overhead for encryption for supporting binary search is only $1\%$.
\vspace{2mm}
\paragraph*{Decryption} The dotted line using the right y-axis in Fig.~\ref{fig:ABEAT_BS} shows the decryption time of \name with linear and binary search under a large intended recipients (vanilla encryption block size). Binary search considerably reduces the time of locating the correct tag. It is approximately $3\%$ to $35\%$ faster than linear search when the intended recipient set is large.

\end{document}